\documentclass[11pt,letterpaper]{article}
\usepackage{relsize}
\usepackage{epigraph}
\usepackage{nicefrac}
\usepackage[margin=1in]{geometry}
\usepackage{hyperref}
\usepackage{colortbl}

\renewcommand{\vec}[1]{\mathbf{#1}}
\renewcommand{\vec}[1]{\mathbf{#1}}
\usepackage{environ}
\NewEnviron{box1}[1]{%
	\begin{center}\fbox{\parbox{6.2in}{%
				{\centering\scshape #1\par}%
				\parskip=1ex
				\everypar{\hangindent=1em}%
				\BODY
}}\end{center}}

\usepackage{geometry}
\usepackage{color, colortbl}
\definecolor{moccasin}{rgb}{0.98, 0.92, 0.84}

\usepackage{enumerate}
\usepackage{amsmath, amsthm, amssymb, amstext, comment, graphicx, fullpage}

\newtheorem{theorem}{Theorem}[section]

\newtheorem{definition}{Definition}

\newtheorem{corollary}{Corollary}

\newtheorem{lemma}{Lemma}

\newtheorem{remark}{Remark}

\usepackage{multirow}
\usepackage{hyperref}
\hypersetup{
	colorlinks   = true, 
	urlcolor     = blue, 
	linkcolor    = blue, 
	citecolor   = blue 
}

\usepackage[dvipsnames]{xcolor}
\colorlet{LightRubineRed}{RubineRed!70!}
\newcommand{\simina}[1]{{\color{blue}\noindent\textbf{ }\marginpar{****}\textit{{#1}}}}%

\title{The Query Complexity of Cake Cutting\footnote{This project has received funding from the European Research Council (ERC) under the European Union’s
		Horizon 2020 research and innovation programme (grant agreement No 740282). 
		The project was also supported by the ISF grant 1435/14 administered by the Israeli Academy of Sciences and Israel-USA Bi-national Science Foundation (BSF) grant 2014389. A part of this work was done while the authors were visiting Microsoft Research.
}}
\author{
	Simina Br\^anzei\footnote{Purdue University, USA. E-mail: \textcolor{blue}{\href{mailto:simina.branzei@gmail.com}{simina.branzei@gmail.com}}.}\\
	\newline
	\and
	Noam Nisan\footnote{Hebrew University of Jerusalem and Microsoft Research, Israel. E-mail: 
	\textcolor{blue}{\href{mailto:noam@cs.huji.ac.il}{noam@cs.huji.ac.il}}.}
}
\date{}

\begin{document}
	\maketitle
\begin{abstract}
We study the query complexity of cake cutting and give lower and upper bounds for computing approximately envy-free, perfect, and equitable allocations with the minimum number of cuts. The lower bounds are tight for computing connected 
envy-free allocations among $n=3$ players and for computing perfect and equitable allocations with minimum number of cuts between $n=2$ players.
We also formalize moving knife procedures and show that a large subclass of this family, which captures all the known moving knife procedures, can be simulated efficiently with arbitrarily small error in the Robertson-Webb query model.
\end{abstract}

\newpage
\section{Introduction}


We study the classical cake cutting problem due to \cite{Steinhaus48}, which captures the division of a heterogeneous resource---such as land, time, mineral deposits, fossil fuels, and clean water \cite{Pro13}---among several parties with equal 
rights but different interests over the resource. This model has inspired a rich body of literature in mathematics, political science, economics \cite{RW98,BT96,Moulin03} and more recently computer science \cite{socialchoice_book} since problems in resource allocation
(and fair division in particular) are central to the design
of multiagent systems \cite{Chevaleyre}. Many protocols for various fair division models are now implemented on the Spliddit website \cite{GP14}.

Mathematically, the problem is to divide the cake, which is the unit interval, among a set of $n$ players with valuation functions induced by probability measures over $[0,1]$. Given such a resource, the goal is to compute an allocation of the cake in which every player is content with the piece received.
A major challenge for a third party trying to compute a desirable allocation is that the preferences of the players are private, and fair allocations can only be found when enough information is on the table. An example of a cake cutting protocol is Cut-and-Choose, which dates back to more than 2500 years ago when it appeared in written records in the context of land division. Cut-and-Choose can be used to obtain an envy-free allocation among two players, such as Alice and Bob trying to cut a birthday cake with different toppings: First Alice cuts the cake in two pieces of equal value to her, then Bob picks his favorite and Alice takes the remainder.

More generally, cake cutting protocols operate in a query model (called Robertson-Webb \cite{WS07}), in which a center that does not know the players asks them questions until it manages to extract enough information about their preferences to determine a fair division.
Two of the prominent fairness notions, envy-freeness and proportionality, have been the subject of in depth study from a computational point of view.
Proportionality requires that each player gets their fair share of the resources, which is their total value for the whole cake divided by the number of players, while envy-freeness is based on social comparison and means no player should want to swap their piece with anyone else's. The two notions are not necessarily comparable, since envy-freeness can be trivially achieved by throwing away the entire resource, which is not true for proportionality. However, when the entire cake must be allocated, envy-freeness implies proportionality and is surprisingly hard to reach. 

The problem of finding an envy-free cake cutting protocol was suggested by \cite{GS58}, and solved by Selfridge and Conway for three players (cca. 1960, see, e.g., \cite{RW98,BT96}) and by 
Brams and Taylor 35 years later \cite{BT95} for any number of players. From a computational point of view, the Brams and Taylor protocol has the major drawback that its runtime can be made arbitrarily long by setting up the valuations of the players appropriately. In 2016, \cite{AM16} announced a breakthrough by giving the first bounded envy-free cake cutting protocol for any number of players, where bounded means that the number of queries is only a function of the number of players, and not of the valuations.

The query complexity of proportionality is well understood. The problem of computing a proportional allocation  with connected pieces can be solved with $O(n \log{n})$ queries in the Robertson-Webb model \cite{EP84}, with a matching lower bound due to \cite{WS07} for connected pieces that was extended to any number of pieces by \cite{EP06b}. 

\smallskip

In contrast, for the query complexity of envy-free cake cutting a lower bound of $\Omega(n^2)$ was given by \cite{Pro09} and an upper bound of $O\left(n^{n^{n^{n^{n^{n}}}}}\right)$ by \cite{AM16}. The bounded algorithm in 
\cite{AM16} 
  outputs highly fragmented allocations, and such resources are generally difficult to handle.
Envy-free allocations with connected pieces do in fact exist for very general valuations \footnote{Such allocations exist even for valuations not induced by probability measures (see, e.g.,  \cite{simmons80}, \cite{Stromquist80} for a proof based on a topological lemma on intersection of sets, and \cite{Su99} for a proof using Sperner's lemma).} (see, e.g., \cite{simmons80,Stromquist80,Su99}), but cannot be computed in the Robertson-Webb model \cite{Stromquist08} except for specific classes such as polynomial functions \cite{sim15}.

In light of this impossibility, it makes sense to study the computation of $\epsilon$-envy-free allocations with few cuts.
For a general utility model where the valuations are arbitrary functions (not necessarily induced by probability measures), bounds on the query complexity of approximate envy-free cake cutting with connected pieces were given in \cite{DQS12} together with a proof of PPAD-hardness. While the problem of computing connected $\epsilon$-envy-free allocations is in PPAD \cite{Su99}, neither the computational hardness nor the query lower and upper bounds   
in \cite{DQS12} carry over (or are close to tight) in the usual cake cutting model, which leaves wide open the question of understanding the complexity of this problem.

Aside from proportionality and envy-freeness, a third notion of fairness is known as equitability, in which each player must receive a piece worth the same value. Equitable and proportional allocations with connected pieces were shown to exist by 
Cechlarova, Dobos, and Pillarova \cite{CDP13}. The computational complexity of approximate equitability was investigated by Cechlarova and Pillarova, who gave an upper bound of $O\left(n \left(\log{n} + \log{\epsilon^{-1}}\right)\right)$ for computing an $\epsilon$-equitable and proportional allocation with connected pieces for any number of players, and by \cite{PW17} who showed a lower bound of $\Omega\left(
\log{\epsilon^{-1}}/\log{ \log \epsilon^{-1}}\right)$ for finding an $\epsilon$-equitable allocation (not necessarily connected) for any number of players.

More stringent fairness requirements are also possible, such as $(i)$ the necklace splitting problem, for which the existence of fair solutions was established by \cite {Neyman46}, with a bound on the number of cuts given by \cite{ALON1987247}, $(ii)$ the more general notion of exact division \footnote {The problem of exact division is the following: given a cake with $n$ players and target non-negative weights $w_1 \ldots w_k$, find a partition $A = (A_1, \ldots, A_k)$ such that $V_i(A_j) = w_j$ for each player $i$ and every piece $A_j$.}, which generalizes necklace splitting and was shown to exist by  \cite{DS61}, and 
$(iii)$ the competitive equilibrium from equal incomes, the existence of which was determined by \cite{Weller}. The necklace splitting problem is contained in PPA and in fact is PPAD-hard \cite{FG18}.
A well known instantiation of the necklace splitting solution is that of perfect allocations, which are simultaneously proportional, envy-free, and equitable.

The complexity and existence of various fairness concepts in cake cutting and related models, such as cake cutting where the resource is a chore, pie cutting, multiple homogeneous goods, multiple discrete goods was studied in \cite{CDGK+17,OPR16,AMNS15,Budish11,GZHK+11,DFHY18,BHM15,Brams2008,CKKK12,PT18,BCEI+17,AFGS+17,BLM16}. 

\subsection{Our Contribution}

In this paper we study the query complexity of cake cutting in the standard Robertson-Webb (RW) query model \cite{WS07} for several fairness notions: envy-free, perfect, and equitable allocations with the minimum number of cuts. Such allocations are known to exist on any instance, but several impossibility results preclude their computation in the standard query model \cite{Stromquist08,CP12}. Nevertheless, the computation of approximately fair solutions is possible and we state a general simulation result implying that for each $\epsilon > 0$, a number of queries that is polynomial in $n$ and $1/\epsilon$ suffices to find $\epsilon$-fair allocations for several fairness concepts, such as $O(n/\epsilon)$ for connected envy-free and $O(n(k-1)/\epsilon)$ for $(\epsilon,k)$-measure splittings.

Our first main contribution is to give lower bounds for the problems of computing approximately fair allocations (with deterministic protocols), showing that finding (i) an $\epsilon$-envy-free allocation with connected pieces among three players, (ii) a connected $\epsilon$-equitable allocation for two players, and (iii) an $\epsilon$-perfect allocation with two cuts for two players requires $\Omega\left(\log{\frac{1}{\epsilon}}\right)$ queries. 
The lower bounds for envy-freeness and perfect allocations are the first query lower bounds for these problems. Also, the lower bound for envy-free allocations extends to any number of players. 
The main idea underpinning these results is that of maintaining a self-reducible structure throughout the execution of a protocol, which may be useful more generally for obtaining other lower bounds.

We also give improved upper bounds for these fairness notions, which in the case of envy-freeness for three players and perfect allocations for two players are $O(\log{\frac{1}{\epsilon}})$ and rely on approximately simulating in the RW model two moving knife procedures, due to \cite{BB04} and \cite{Austin82}, respectively. An upper bound of 
$O\left(\log{\frac{1}{\epsilon}} \right)$ for computing a connected $\epsilon$-equitable allocation for two players was given by \cite{CP12}. The upper bounds do not assume any restrictions on the valuations.

A summary of our bounds for the query complexity of computing fair allocations can be found in Table 1, together with those from previous literature. 


\begin{table}[h!] 
	\label{tab:summary}
	\begin{center}
		\begin{tabular}{||c | c|c|c|c ||} 
			\hline \hline
			Fairness notion & Players & Upper bound & Lower bound\\
			\hline \hline
			\multirow{3}{*}{\centering{$\epsilon$-envy-free (connected)}}& $n=2$ & 1 & 1 \\
			& $n = 3$ & $O(\log{\epsilon^{-1}})$ (\textcolor{red}{$*$}) & $\Omega(\log{\epsilon^{-1}})$ (\textcolor{red}{$*$}) \\ 
			& $n \geq 4$ & $O(n/\epsilon)$ (\textcolor{red}{$*$}) & $\Omega\left(\log{\epsilon^{-1}}\right)$  (\textcolor{red}{$*$})\\
			\hline
			\multirow{2}{*}{\centering{$\epsilon$-perfect (min  cuts)}}& $n=2$ & $O(\log{\epsilon^{-1}})$ (\textcolor{red}{$*$}) & $\Omega(\log{\epsilon^{-1}})$ (\textcolor{red}{$*$}) \\ 
			& $n \geq 3$ & $O\left(n^3/\epsilon\right)$ \cite{BM15} & $\Omega\left(
			\frac{\log{\epsilon^{-1}}}{\log{ \log \epsilon^{-1}}}\right)$ \cite{PW17} \\ \hline
			\multirow{2}{*}{\centering{$\epsilon$-equitable (connected)}}& $n=2$ & $O(\log{\epsilon^{-1}})$ \cite{CP12} & $\Omega(\log{\epsilon^{-1}})$ (\textcolor{red}{$*$}) \\
			& $n \geq 3$ & $O\left(n \left(\log{n} + \log{\epsilon^{-1}}\right)\right)$ \cite{CP12} & $\Omega\left(
			\frac{\log{\epsilon^{-1}}}{\log{ \log \epsilon^{-1}}}\right)$ \cite{PW17} \\ \hline \hline 
			envy-free (exact) & $n \geq 2$ &  $O\left(n^{n^{n^{n^{n^{n}}}}}\right)$ \cite{AM16} & $\Omega(n^2)$ \cite{Pro09} \\ \hline
			proportional (exact) & $n \geq 2$ & $O\left(n \log{n}\right)$ \cite{EP84} & $\Omega(n \log{n})$ \cite{WS07,EP06b} \\ \hline \hline
		\end{tabular}
		\caption{Query complexity in cake cutting in the standard query model. Our results are marked with (\textcolor{red}{$*$}).
			The lower bounds for finding $\epsilon$-perfect and $\epsilon$-equitable allocations for $n\geq 3$ players hold for any number of cuts \cite{PW17}. The bounds for exact envy-free and proportional allocations hold for any number of cuts, except the upper bound for proportional works for connected pieces.}
	\end{center}
\end{table}



Our second main contribution is to formalize the class of moving knife procedures, which was previously viewed as disjoint from the RW query model. Using this definition, we show that any fair moving knife procedure with a fixed number of players and devices can be simulated in $O(\log{\frac{1}{\epsilon}})$ queries in the RW model when the players have value densities that are bounded from above by a constant. 
This simulation immediately implies that all the known moving knife procedures, such as the procedures designed by Austin, Barbanel-Brams, Webb \cite{Webb}, and Brams-Taylor-Zwicker \cite{BTZ95}, can be simulated efficiently within $\epsilon$-error in the RW model when the measures of the players are bounded. 
In this context we also give a moving knife procedure for computing connected equitable allocations for any number of players.

\section{Model}


The resource (cake) is represented as the interval $[0,1]$. There is a set of players $N=\{1,\ldots,n\}$ interested in the cake, such that each player $i\in N$ is endowed with a private \emph{valuation function} $V_i$ that assigns a value to every subinterval of $[0,1]$. These values are induced by a non-negative integrable \emph{value density function} $v_i$, so that for an interval $I$, $V_i(I)=\int_{x\in I} v_i(x)\ dx$. The valuations are additive, so $V_i\left(\bigcup_{j=1}^m I_j\right) = \sum_{j=1}^m V_i(I_j)$ for any disjoint intervals $I_1, \ldots, I_m \subseteq [0,1]$. The value densities are non-atomic, and in fact sets of measure zero are worth zero to a player.
Without loss of generality, the valuations are normalized to $V_i([0,1]) = 1$, for all $i = 1 \ldots n$. 

A \emph{piece of cake} is a finite union of disjoint intervals. A piece is \emph{connected} if it consists of a single interval.
An \emph{allocation} $A = (A_1, \ldots, A_n)$ is a partition of the cake among the players, such that each player $i$ receives the piece $A_i$, the pieces are
disjoint, and $\bigcup_{i \in N} A_i = [0,1]$.

We say that a value density $v$ is: $(i)$ \emph{hungry} if $v(x) > 0$ for all $x \in [0,1]$, $(ii)$ \emph{uniform} on an interval $[a, b]$ if $v(x) = 1$ for all $x \in [a,b]$, and $(iii)$ \emph{bounded} on interval $[a,b]$ if there exists a constant $D > 0$ such that $v(x) \geq D$ for all $x \in [a,b]$ (or simply bounded if $a = 0$ and $b=1$). We will sometimes to refer to a valuation (or player) as hungry to simply mean that the corresponding density is strictly positive.

\subsection{Fairness Notions}

An allocation $A$ is said to be \emph{proportional} if $V_i(A_i) \geq 1/n$ for all $i \in N$; \emph{envy-free} if $V_i(A_i) \geq V_i(A_j)$ for all $i,j \in N$; \emph{perfect} if $V_i(A_j) = 1/n$ for all $i,j \in N$; \emph{equitable} if $V_i(A_i) = c$, for all $i \in N$ and some value $c \in [0,1]$.
In addition to allocations, we will refer to partitions as divisions of the cake into pieces where the number of pieces need not equal the number of players. A partition $A = (A_1, \ldots, A_k)$ is said to be a $k$-measure splitting if $V_i(A_j) = 1/k$ for each player $i$ and piece $A_j$. 

Envy-free allocations with connected pieces and equitable allocations with connected pieces always exist \cite{Su99,CDP13}, while $k$-measure splittings exist with $n(k-1)$ cuts \cite{ALON1987247}.

We will also be interested in $\epsilon$-fair division, where the fairness constraints are satisfied within $\epsilon$-error; for instance, an allocation $A$ is $\epsilon$-envy-free if $V_i(A_i) \geq V_i(A_j) - \epsilon$, for each $i, j = 1 \ldots n$, and an $(\epsilon,k)$ measure splitting if $V_i(A_j) \in (1/k - \epsilon, 1/k + \epsilon)$ for each $i = 1 \ldots n$, $j = 1 \ldots k$.

More generally, we consider an abstract definition for fairness notions that admit approximate versions as follows.

\begin{definition}[Abstract fairness]
	A fairness notion $\mathcal{F}$ must satisfy the following condition:
	\begin{itemize}
		\item Let $A = (A_1, \ldots, A_n)$ be any allocation that is $\mathcal{F}$-fair with respect to some valuations $\vec{v} = (v_1 \ldots v_n)$. Then for any $\epsilon > 0$ and valuations $\vec{v}' = (v_1' \ldots v_n')$, if $|V_i(A_j) - V_i'(A_j)| \leq \epsilon/2$, $\forall i,j = 1 \ldots n$, then it must be the case that the allocation $A$ is $\epsilon$-$\mathcal{F}$-fair with respect to valuations $\vec{v}'$.
	\end{itemize}
\end{definition}
Proportionality, envy-freeness, and perfection are instantiations of this notion of abstract fairness. For example, if an allocation $A$ is envy-free with respect to some valuations $\vec{v}$, then $A$ is also $\epsilon$-envy-free for valuations $\vec{v}'$ that are close to $\vec{v}$.

\subsection{Query Complexity}
All the discrete cake cutting protocols operate in a query model known as the Robertson-Webb model (see, e.g., the book of \cite{RW98}), which was explicitly stated by~\cite{WS07}. In this model, the protocol communicates with the players using the following types of queries:
\begin{itemize}
	\item{}$\emph{\textbf{Cut}}_i(\alpha)$: Player $i$ cuts the cake at a point $y$ where $V_{i}([0,y]) = \alpha$, where $\alpha \in [0,1]$ is chosen arbitrarily by the center \footnote{Ties are resolved deterministically, using for example the leftmost point with this property.}. The point $y$ becomes a {\em cut point}.
	\item{}$\emph{\textbf{Eval}}_i(y)$:  Player $i$ returns $V_{i}([0,y])$, where $y$ is a previously made cut point.
\end{itemize}

An RW protocol asks the players a sequence of cut and evaluate queries, at the end of which it outputs an allocation demarcated by cut points from its execution (i.e. cuts discovered through queries).
Note that the value of a piece $[x,y]$ can be determined with two Eval queries, $Eval_i(x)$ and $Eval_i(y)$.

A second class of protocols is known as ``moving-knife'' (or continuous) procedures, which typically involve sliding multiple knives across the cake, while evaluating the players' valuations until some stopping condition is met. This class has not been formalized until now. Examples of such procedures include Austin's procedure, which computes a perfect allocation for two players, Stromquist's procedure, for finding a connected envy-free allocation for three players, and Dubins-Spanier, for computing a proportional allocation for any number of players. The latter is the only moving knife procedure that can be simulated exactly in the RW model.

\section{Simulation of Partitions} \label{sec:gen}

A general technique useful for computing approximately fair allocations in the RW model is based on asking the players to submit a discretization of the cake in many small cells via Cut queries, and reassembling them offline in a way that satisfies approximately the desired solution. The next statement implies that it is possible to approximate in the RW model allocations that exist with a bounded number of cuts, which implies upper bounds for computing approximately fair solutions for several fairness concepts; the proof can be found in Appendix \ref{app:sim_partition}. The difference from \cite{BBKP14}, where discretizations are used to simulate bounded algorithms, is that here we do not restrict ourselves to properties that are computable by an RW protocol \footnote{In fact, in our case a solution with a required property as in \cite{BBKP14} need not even exist on all instances.}.

\begin{lemma} \label{lem:simgen}
Consider a cake cutting problem for $n$ players.
Suppose there exists a partition $A = (A_1, \ldots, A_m)$, where each piece $A_j$ is demarcated with at most $K$ cuts \footnote{A connected piece $X =[a,b]$ is demarcated by two cuts (namely its endpoints $a$, $b$), a piece $X' = [a,b] \cup [c,d]$ where $a < b < c < d$ is demarcated with four cuts ($a,b,c,d$), and so on. We are interested in the minimum number of points that can be used to demarcate a piece.} and worth $w_{i,j}$ to each player $i$. Then for all $\epsilon > 0$, a partition 
$\tilde{A} = (\tilde{A}_1, \ldots, \tilde{A}_m)$ where each piece $A_j$ is demarcated with at most $K$ cuts and $|V_i(\tilde{A}_j) - w_{i,j}| \leq \epsilon$ for all $i,j$ can be found with $O(K \cdot n/\epsilon)$ queries.
\end{lemma}

As a corollary, we get bounds for computing approximately fair allocations in the RW model. 


\begin{corollary} \label{cor:apx_general}
	For each $\epsilon > 0$ and number of players $n$ with arbitrary value densities, a partition that is
	\begin{itemize}
		\item $\epsilon$-envy-free and connected can be computed with $O(n/\epsilon)$ queries
		\item $\epsilon$-equitable, proportional, and connected can be computed with $O(n/\epsilon)$ queries
		\item an $(\epsilon,k)$-measure splitting with $k$ pieces can be computed with $O(k\cdot n^2/\epsilon)$ queries.
		\item $\epsilon$-perfect can be computed with $O(n^3/\epsilon)$ queries.
	\end{itemize} 
\end{corollary}


\section{Envy-Free Allocations}

Exact connected envy-free allocations are guaranteed to exist (see, e.g., 
~\cite{Stromquist80}, Su~\cite{Su99}), but cannot be computed by finite RW protocols (~\cite{Stromquist08}). However, the impossibility result of Stromquist does not explain how many queries are needed to obtain connected $\epsilon$-envy-free allocations.

From Corollary \ref{cor:apx_general}, a connected $\epsilon$-envy-free allocation can be computed with $O(n/\epsilon)$ queries.
As we show next, fewer queries are needed for three players. The proof relies on approximately computing the outcome of a moving knife procedure due to ~\cite{BB04}. The proof must handle the fact that a discrete RW protocol does not allow the center to cut the cake directly, so the simulation has to search for the position of the ``sword'' in the Barbanel-Brams protocol by keeping track of a small interval that must be made very small from the perspective of each player. The details can be found in Appendix \ref{app:ef} together with all the omitted proofs of this section.

\begin{theorem} \label{thm:ef3sim}
	A connected $\epsilon$-envy-free allocation for three players can be computed with $O\left(\log{\frac{1}{\epsilon}}\right)$ queries.
\end{theorem}

Our main result in this section is a matching lower bound for this problem.

\begin{theorem} \label{thm:3lb}
	Computing a connected $\epsilon$-envy-free allocation for three players requires $\Omega\left(\log \frac{1}{\epsilon}\right)$ queries.
\end{theorem}

We note that in general there are many envy-free allocations so the proof must ensure that for a long enough time none of these solutions are near. We will start from a family of valuations known as ``rigid measure systems'' \cite{Stromquist08}, which were studied in the context of showing that no exact connected envy-free allocation can be computed in general in the RW model.
The difficulty is that while in \cite{Stromquist08} it was sufficient to avoid a single point in order for the protocol to not be able to find an exact envy-free solution with one more query, here we must avoid an entire interval and fit the valuations accordingly in order for an approximately envy-free solution to remain far away after one more query. Towards this end, we will make a self reducible structure from a generalized version of this class of valuations, which will then be used to give the required lower bound.
We first slightly generalize rigid measure systems to allow a different parameter $t_i$ for each player $i$.

\begin{definition} \label{def:rms} 
	A tuple of value densities $(v_1, \ldots, v_n)$ is a generalized rigid measure system if:
	\begin{itemize}
		\item the density of each measure is bounded by: $\frac{1}{\sqrt{2}} < v_i(x) < \sqrt{2}$, for all $i =1 \ldots n$ and $x \in [0,1]$.
		\item there exist points $0 = x_0 < x_1< \ldots< x_{n-1}< x_n =1$ and values $s_i,t_i$ for each player $i$ such that $0 < s_i < 1/n < t_i < 1/2$ and the valuations of the players satisfy the constraints: $V_j(x_{j-1}, x_j) = V_j(x_j, x_{j+1}) = t_j$ for all $j = 1 \ldots n-1$ and $V_n(x_{n-1}, x_n)= V_n(0,x_1)=t_n$. 
	\end{itemize}
\end{definition}

An example can be found in Appendix \ref{app:ef}. Generalized rigid measure systems satisfy the property that the valuations of the players for any given piece cannot differ too much. 

\begin{lemma} \label{lem:density}
	Consider any cake cutting problem where for two players $i$ and $j$ 
	where there exist $a,b>0$ such that for all $x \in [0,1]$, $1/a < v_i(x) < b$ and $1/a < v_j(x)  < b$.
	Then for any two pieces $S_1,S_2$ of the cake, if
	$V_i(S_1) \geq ab \cdot V_i(S_2)$, it follows that $ab \cdot V_j(S_1) > V_j(S_2)$.
\end{lemma}

A useful notion to measure how close a protocol is to discovering an approximately envy-free solution on a given instance will be that of a partial rigid measure system, which we define for three players; the $n>3$ definition works similarly.

\begin{definition} \label{def:prms} 
	A tuple of value densities $(v_1, v_2, v_3)$ is a partial rigid measure system if
	\begin{itemize}
		\item the density of each player $i$ is bounded everywhere: $\frac{1}{\sqrt{2}} < v_i(x) < \sqrt{2}$, for all $x \in [0,1]$.
		\item there exist values $k > 0$ and $1/2 > \ell_i > 1/3 > m_i > 0$ for each player $i$, and points $x,x',y,y' \in [0,1]$, so that the matrix of valuations for pieces demarcated by these points is:
		\begin{table}[h!]
			\label{tab:prms}
			\begin{center}
				\begin{tabular}{l | c | c | c | c | c r}
					
					& $[0,x]$ & $[x,x']$ & $[x',y]$ & $[y,y']$ & $[y',1]$ \\ \hline
					$V_1$ & $\ell_1$ & $k$ & $\ell_1$ & $k$ & $m_1$ \\
					$V_2$ & $m_2$ & $k$ & $\ell_2$ & $k$ & $\ell_2$ \\
					$V_3$ & $\ell_3$ & $k$ & $m_3$ & $k$ & $\ell_3$ 
				\end{tabular}
			\end{center}
		\end{table}
	\end{itemize}
\end{definition}
\vspace{-8mm}

Partial rigid measure systems have the property that if there is a collection of cut points $\mathcal{P} = \{z_1, \ldots, z_k\} \subset [0,1]$, such that there are no points from $\mathcal{P}$ in the intervals $(x,x')$ and $(y,y')$, then any connected partition attainable using cut points from $\mathcal{P}$ has envy at least $0.01 k$.

\begin{lemma} \label{lem:far_ef3}
	Let $(v_1, v_2, v_3)$ be a partial rigid measure system with parameters $k, m_i, \ell_i$ so that for each $i$, $V_i([x,x']) = V_i([y, y']) = k$ for  some $x,x',y,y'$. 
If $\mathcal{P} = \{z_1, \ldots, z_{t}\} \subset [0,1]$ is a collection of cut points such that $x,x',y,y' \in \mathcal{P}$ and there are no points from $\mathcal{P}$ in the intervals $(x,x')$ and $(y,y')$, then any connected partition demarcated by points in $\mathcal{P}$ has envy at least $0.01 k$.
\end{lemma}

Next we show how queries can be answered one at a time so that the valuations remain consistent with (some) partial rigid measure system throughout the execution of a protocol.

\begin{lemma} \label{lem:hide}
	Suppose that at some point during the execution of an $RW$ protocol for three players the valuations and cuts discovered are consistent with a partial rigid measure system with parameters $k > 0$, $0 < m_i < 1/3 < \ell_i < 1/2$ for each $i$, and points $x,y$, so that the valuations are:
\begin{figure}[h!] 
	\centering
	\includegraphics[scale=0.85]{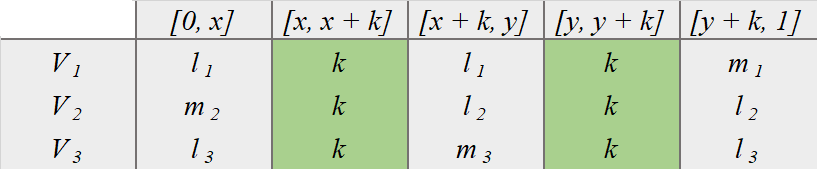}
	\label{fig:partial_main}
\end{figure}

If the intervals $I = [x,x+k]$ and $J = [y,y+k]$ have no cut points inside, then a new cut query can be answered so that the valuations remain
	consistent with a partial rigid measure system where two new intervals $I' \subseteq I$ and $J' \subseteq J$ have no cuts inside, length $0.01k$, and the densities of all the players are uniform on $I',J'$.
\end{lemma}

\begin{proof}(of Theorem \ref{thm:3lb})
	Set the initial configuration to a partial rigid measure system with $k = 0.01$, $\ell_i = 0.35$,
	$m_i=0.28$ for each player $i$. Declare initial cuts at $0.34, 0.35, 0.67, 0.68$ and set the intervals $I = [0.34,0.35]$ and $J = [0.67,0.68]$ (see Table 2 in Appendix \ref{app:ef}). It can be verified there exist compatible valuations for which the densities are in $(1/\sqrt{2}, \sqrt{2})$.


	By iteratively applying Lemma \ref{lem:hide} with every Cut query, a protocol discovers with every cut query a partial rigid system, where the intervals $I$ and $J$ always have uniform density, and their length cannot be diminished by a factor larger than 100 in each iteration.
	By Lemma \ref{lem:density}, if a protocol encounters a partial rigid measure system for which there are no cuts inside $I$ and $J$, where $|I| = |J| = k$, then any configuration attainable with the existing cuts leads to envy of at least $0.01k$.
	To get $\epsilon$-envy, we need $k/100 < \epsilon$, and so the number of queries is $\Omega\left( \log{\frac{1}{\epsilon}}\right)$. 
\end{proof}

The construction can be extended to give a lower bound for any number of players.

\begin{theorem} \label{thm:lb_envyfree_any}
	Computing a connected $\epsilon$-envy-free allocation for $n\geq 3$ players requires $\Omega \left(\log{\frac{1}{\epsilon}}\right)$ queries.
\end{theorem}

The lower bound of $\Omega \left(\log{\frac{1}{\epsilon}}\right)$ is in fact tight for the class of generalized rigid measure systems, for any (fixed) number of players. We show an upper bound of $O\left( \log{\frac{1}{\epsilon}}\right)$ for this class by designing a moving knife procedure and then simulating it in the discrete RW model.

\begin{theorem}  \label{thm:rms_sim}
	For the class of generalized rigid measure systems, a connected $\epsilon$-envy-free allocation  can be computed with $O\left( \log{\frac{1}{\epsilon}}\right)$ queries for any fixed number $n$ of players.
\end{theorem}
\section{Perfect Allocations}

As mentioned in Corollary \ref{cor:apx_general}, $\epsilon$-perfect allocations with the minimum number of cuts can be  
computed with $O(n^3/\epsilon)$ queries. 
For $n=2$ players, the problem of computing $\epsilon$-perfect allocations can be solved 
more efficiently by simulating Austin's moving procedure in the RW model. The proofs for this section are in Appendix \ref{apdx:perfect}.

\begin{theorem} \label{thm:perfect_ub}
	An $\epsilon$-perfect allocation for two players can be computed with $O(\log{\frac{1}{\epsilon}})$ queries.
\end{theorem}

As we show next, this bound is optimal.

\begin{theorem} \label{thm:perfect_lb}
	Computing an $\epsilon$-perfect allocation with the minimum number of cuts for two players requires $\Omega\left(\log\frac{1}{\epsilon}\right)$ queries.
\end{theorem}

We prove the lower bound by maintaining throughout the execution of any protocol two intervals in which the cuts of the perfect allocation must be situated, such that the distance to a perfect partition cannot decrease too much with any cut query.


\begin{lemma} \label{lem:perfect_induction}
	Let $\epsilon > 0$. Consider a two player instance consistent with Figure \ref{fig:perfect_main}, where
	\begin{description}
		\item[\hspace{4mm}\emph{1.}] $\epsilon < 0.001 \min\{a,d\}$.
		\item[\hspace{4mm}\emph{2.}] any allocation obtained with cuts $0 < k < \ell < 1$ that is $\epsilon$-perfect from the point of view of player $1$ is worth to player $2$ less than $0.5 - d/100-\epsilon$ when $k < x$ and more than $0.5 + d/100 + \epsilon$ when $k > x+a$.
		\item[\hspace{4mm}\emph{3.}] there are no cut points inside the intervals $I = [x,x+a]$ and $J = [y,y+a]$.
	\end{description} 
	\begin{figure}[h!]
			\centering
		\includegraphics[scale=0.85]{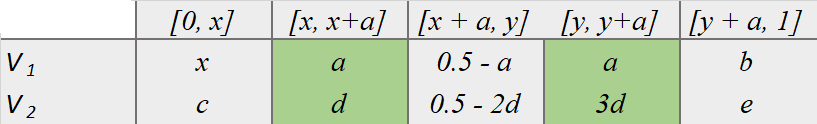}
	\caption{Construction for perfect. Player $1$ has uniform density everywhere; $y= x+0.5$, $0 < a,d \leq 0.1$, $x,b,c,e > 0$, $x + a + b = 0.5$ and $c + 2d + e = 0.5$.}
		\label{fig:perfect_main}
	\end{figure}

	Then a new query can be handled so that the valuations remain consistent with Figure \ref{fig:perfect_main}, such that conditions $2$ and $3$ still hold with respect to intervals $I' = [x',x'+a'],J' = [y',y'+a']$ and parameters $x',a'=a/100$, and $d'=d/100$.
\end{lemma}

By iteratively applying Lemma~\ref{lem:perfect_induction} with every cut query received from a suitably chosen starting configuration, we will obtain that the number of rounds is $\Omega \left( \log{\frac{1}{\epsilon}}\right)$, implying Theorem \ref{thm:perfect_lb}.

\section{Equitable Allocations}

Cechlarova, Dobos, and Pillarova~\cite{CDP13} showed that for any number of players and any order, there exists a connected equitable allocation in that order. Moreover, the equitable allocation is proportional for some order. We give a tight lower bound on the number of queries required for finding connected $\epsilon$-equitable allocations; an upper bound was given in \cite{CP12}. 

\begin{theorem}[\cite{CP12}]
For any fixed number $n$ of players, a connected $\epsilon$-equitable and proportional allocation can be computed with $O(\log{\frac{1}{\epsilon}})$ queries.
\end{theorem}

\begin{theorem} \label{thm:equitable_lb}
	Computing a connected $\epsilon$-equitable allocation for two players requires $\Omega\left(\log\frac{1}{\epsilon}\right)$ queries.
\end{theorem}

For two hungry players, the connected equitable and proportional allocation is unique.

\begin{lemma} \label{lem:ep_unique}
	For two players with hungry valuations, the cut point of the equitable allocation is unique.
\end{lemma}
\begin{proof}
	A unique equitable allocation exists for each order of the players \cite{CDP13}. Let $x$ be the cut point of  the equitable allocation when the player order is $(1,2)$. Then there exists $c$ such that $V_1(0,x) = V_2(x,1) = c$, and so $V_2(0,x) = V_1(x,1) = 1-c$. Thus the cut point of the equitable allocation is the same for each order of the players.
\end{proof}

Next we show that when the valuations of the players are as in the next table and no cuts may be used from the interval $(x,y)$, the distance from a connected equitable allocation is high.

		\begin{figure}[h!]
	\centering
	\includegraphics[scale=0.85]{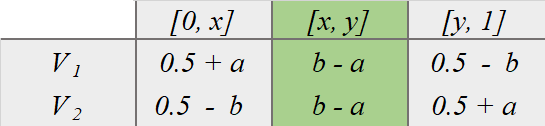}
	\caption{Construction for equitable lower bound. The distance from a connected equitable and proportional allocation is $b-a$, where $0 < a < b < 0.5$ and $0 < x < y < 1$.}
	\label{fig:equitable_main}
\end{figure}

\begin{lemma} \label{lem:ep_distance}
	Consider a two player problem where there exist points $0 < x < y < 1$ and values $0 < a < b < 0.5$ such that the valuations are consistent with Figure \ref{fig:equitable_main}, where $V_1(0,x) = 0.5+a = V_2(y,1)$, $V_2(x,1) = 0.5+b = V_1(0,y)$.
	 Then any connected allocation that can be formed with cut points outside the interval $(x,y)$ has distance at least $b-a$ from equitability.
\end{lemma}
\begin{proof}
	Consider first allocations that can be obtained by cutting at $x$. If the player order is $(1,2)$, then $|V_1(0,x) - V_2(x,1)| = |(0.5+a) - (0.5+b)| = b-a$. Moreover, for any point $0 < z < x$, the order $(1,2)$ gives player $1$ less than $w+a$ and player $2$ more than $w+b$, which leads to distance at least $b-a$ from equitability. On the other hand, if the player order is $(2,1)$, then $V_2(0,x) = 0.5-b$, while $V_1(x,1) = 0.5-a$. This allocation has distance $|V_2(0,x) - V_1(x,1)| = b -a $ from equitability. The distance only increases for any cut $z < x$, since $|V_2(0,z) - V_1(z,1)| = V_1(z,1) - V_2(0,z) > (0.5-a) - (0.5-b) = b - a$.
	
	Similarly, if the cut is at $y$, then the player order $(1,2)$ gives distance $|V_1(0,y) - V_2(y,1)| = |(0.5+b) - (0.5+a)| = b-a$. For any cut point $z > y$, the distance only increases, since $a<b$ and player $1$ gets more than $0.5+b$, while player $2$ gets less than $0.5+a$. Finally, if the player order is $(2,1)$, the cut point $y$ leads to distance $|V_2(0,y) - V_1(y,1)| = |(0.5-a) - (0.5-b)| = b -a$. For any cut point $z > y$, $|V_2(0,z) - V_1(z,1)| = V_2(0,z) - V_1(z,1) > (0.5-a) - (0.5-b) = b-a$.
\end{proof}

Given such a configuration, queries can be handled in a way that preserves the symmetry and the distance to equitability gets reduced by a constant factor.
\begin{lemma} \label{lem:ep_reduce}
	Consider a two player problem where there exist points $0 < x < y < 1$ and values $0 < a < b < 0.5$ such that $V_1(0,x) = 0.5+a = V_2(y,1)$, $V_2(x,1) = 0.5+b = V_1(0,y)$. Then any Cut query (addressed to either player) can be answered so that the new configuration has two new points $z < t$ such that $z,t \in (x,y)$, the valuations satisfy $V_1(0,z) = 0.5+a'=V_2(t,1)$, $V_2(z,1)= 0.5+b'=V_1(0,t)$, and $b' - a' \geq (b-a)/100$.
\end{lemma}
\begin{proof}
	Suppose that player $1$ receives a cut query $Cut_1(\alpha)$. If $\alpha < 0.5+a$ or $\alpha > 0.5+b$, then answer for both players in a way that is consistent with the history (e.g. uniform on the interval where the unique point determined by the answer of player $1$ to the query falls). Otherwise, define new cuts $z,t$ such that $x < z < t < y$ and consider two subcases: 
	\begin{itemize}
		\item $\alpha \in \left(0.5 +a, 0.5 + (a+b)/2\right.\left.\right]$.
		Let $a' = (a+b)/2 + (b-a)/100$ and $b' = (a+b)/2 + (b-a)/50$. Then $a < (a+b)/2 < a' < b' < b$. 
		
		\item $\alpha \in \left(0.5 +  (a+b)/2, \alpha < 0.5 + b\right)$. Let $a' = (a+b)/2 - (b-a)/50$, $b' = (a+b)/2 - (b-a)/100$. Then $a < a' < b' < (a+b)/2 < b$. 
	\end{itemize} 
	In both cases, $b' - a' = (b-a)/100$. Set $V_1(0,z) = 0.5 + a'$, $V_2(0,z) = 0.5 - b'$, $V_1(0,t) = 0.5 + b'$, and $V_2(0,t) = 0.5-a'$. In the first case, the answer to the query falls to the left of $z$ and the value can be set in any way consistent with the total value of player $2$ for $[x,z]$, while in the second case the cut point falls to the right of $t$ and can similarly be handled arbitrarily on $[t,y]$. Afterwards, fit the value of player $1$ for the generated cut point in a way that is consistent with its valuation for $[0,z]$ and $[0,t]$.
	
	If player $2$ receives instead a cut query $Cut_2(\alpha)$, then when $\alpha < 0.5 - b$ or $\alpha > 0.5-a$, the query can be answered arbitrarily in a way that is consistent with the history. Otherwise, define cuts $z,t$ such that $x < z < t < y$ and consider the subcases:
	\begin{itemize}
		\item $\alpha \in \left(0.5 - b, 0.5 - (a+b)/2\right.\left.\right]$.
		Let $a' = (a+b)/2 - (b-a)/50$ and $b' = (a+b)/2 - (b-a)/100$. Then $a < a' < b' < (a+b)/2 < b$ and $0.5 - (a+b)/2 < 0.5 -b'$.
		
		\item $\alpha \in \left(0.5 - (a+b)/2, \alpha < 0.5 - a\right)$. Consider values $a' = (a+b)/2 + (b-a)/100$, $b' = (a+b)/2 + (b-a)/50$. Then $a < (a+b)/2 < a' < b' < b$ and $0.5 - a' < 0.5 - (a+b)/2$.
	\end{itemize}
	We have $b'-a' = (b-a)/100$ and the valuations of players $1$ and $2$ for the pieces $[0,z]$ and $[0,t]$ can be defined as in case 1, when player $1$ received the Cut query.
\end{proof}

The proof of Theorem \ref{thm:equitable_lb} follows from the previous lemmas. 

\begin{proof} (of Theorem \ref{thm:equitable_lb})
	Start with cut points $x = 0.4$, $y = 0.6$, and values $a = 0.05$, $b = 0.06$, such that $V_1(0,x) = 0.55$, $V_2(0,x) = 0.44$, $V_1(0,y) = 0.56$, $V_2(0,y) = 0.45$. By Lemma \ref{lem:ep_distance}, the distance to an equitable and proportional allocation by using cuts outside $(x,y)$ is at least $b-a = 0.01$. By applying Lemma \ref{lem:ep_reduce} with every Cut query received, we get that the distance is reduced by a factor of 100 in every round. For $\epsilon$-equitability to hold in round $k$, the condition $0.01/100^k \leq \epsilon$ must be met, and so the number of rounds is $\Omega\left(\log \frac{1}{\epsilon} \right)$.
\end{proof}

\section{Moving Knife Protocols} \label{sec:movingknife}

We will consider a family of protocols that seems to, on one hand, capture all types of 
protocols that have so far been called ``moving knife'' procedures and, on the other hand, be
simple enough for a transparent simulation. An important ingredient of the definition is that knife positions must be continuous. To ensure that ``cut queries" fall within the definition,
we will only require continuity for hungry valuation functions. The formal definition together with the omitted proofs of this section can be found in Appendix \ref{app:movingknife}.

\begin{theorem} \label{thm:mainsimulation_robust}
	Consider a cake cutting problem where the value densities are bounded from above and below by strictly positive constants.
	Let $\mathcal{M}$ be an RW moving knife protocol with at most $r$ steps, such that $\mathcal{M}$ outputs $\mathcal{F}$-fair allocations demarcated by at most a constant number $C$ of cuts.
	
	Then for each $\epsilon > 0$, there is an RW protocol $\mathcal{M}_{\epsilon}$ that uses $O\left(\log \frac{1}{\epsilon}\right)$ queries and computes $\epsilon$-$\mathcal{F}$-fair partitions demarcated with at most $C$ cuts.
\end{theorem}

\begin{theorem}\label{thm:mainsimulation_all}
	The Austin, Austin's extension, Barbanel-Brams, Stromquist, Webb, Brams-Taylor-Zwicker, and Saberi-Wang moving knife procedures
	can be simulated with $O\left(\log{\frac{1}{ \epsilon}}\right)$ RW queries when the value densities are 
	bounded from above and below by positive constants.
\end{theorem}

\noindent \textbf{An Equitable Protocol}:
Next we show a simple moving knife protocol in the Robertson-Webb model for computing equitable allocations for any number of hungry players.
A more complex moving knife procedure for computing exact equitable allocations that is not in the RW model but works even when the valuations are not hungry was discovered independently by \cite{SS17}.

\bigskip


	\emph{\textbf{Equitable Protocol}} : 
	{\em Player $1$ slides a knife continuously across the cake, from $0$ to $1$. 
		For each position $x_1$ of the knife, player $1$ is asked for its value of the piece $[0,x_1]$; then 
		each player $i = 2 \ldots n$ iteratively positions its own knife at a point $x_i \in [x_{i-1}, 1]$ with $V_i(x_{i-1}, x_{i})$ $=$ $V_1(0,x_1)$ if possible, and at $x_i = 1$ otherwise. 
		
		Player $n$ shouts ``Stop!" when its own knife reaches the right endpoint of the cake (i.e., $x_n = 1$).
		The cake is allocated in the order $1 \ldots n$, with cuts at $x_1 \ldots x_{n-1}$.
	}

\bigskip

\begin{theorem} \label{thm:equit_step}
	There is is an RW moving knife protocol that computes a connected equitable allocation for any number $n$ of hungry players.
\end{theorem}

This immediately implies a moving knife protocol for computing an allocation that is not only equitable, but also proportional; this can be achieved by running Equitable Protocol for every permutation of the players and choosing the one that is proportional.

\section{The Stronger and Weaker Models}
We also discuss two other query models; the proofs for this section are in Appendix \ref{app:weakstrong}.
\medskip

The first one, which we call $RW^+$, is stronger in that the inputs to evaluate queries need not be previous cut points and at the end the protocol can use arbitrary points (i.e. not just cuts discovered through queries) to demarcate the final allocation.

\begin{definition}[$RW^+$ query model]
	An $RW^{+}$ protocol for cake cutting communicates with the players via two types of queries:
	\begin{itemize}
		\item{}$\emph{\textbf{Cut}}_i(\alpha)$: Player $i$ cuts the cake at a point $y$ where $V_{i}([0,y]) = \alpha$, for any $\alpha \in [0,1]$.
		\item{}$\emph{\textbf{Eval}}_i(y)$: Player $i$ returns $V_{i}([0,y])$, for any $y \in [0,1]$.
	\end{itemize}
	At the end of execution an $RW^{+}$ protocol outputs an allocation that can be demarcated by any points of its choice, regardless of whether they have been discovered through queries or not.
\end{definition}

The $RW^+$ model differs from the $RW$ model in subtle ways. For instance, in $RW$ there exists a characterization of truthful protocols (i.e. all truthful protocols are dictatorships for $n=2$ players, with a similar statement for $n \geq 3$ players \cite{BM15}). A similar characterization is not known in the $RW^+$ model. However the $RW^+$ model allows a general simulation of moving knife protocols without requiring that valuations are bounded from below. Our constructions for the lower bounds work against this stronger query model. 

Clearly the discretization technique in Section \ref{sec:gen} still applies, and so do the upper bounds from there for $RW^+$ protocols. We get that in fact the lower bounds hold too since our constructions did not use in any crucial way the fact that the evaluate inputs must come from previous cut queries.

\begin{corollary}
	Computing a connected $\epsilon$-envy-free allocation among $n=3$ players, an $\epsilon$-perfect allocation with two cuts between $n=2$ players, and a connected $\epsilon$-equitable allocation between $n=2$ players in the $RW^+$ models requires $\Theta(\log{\frac{1}{\epsilon}})$ queries. Computing a connected $\epsilon$-envy-free allocation for any fixed number $n$ of players requires $\Omega(\log\frac{1}{\epsilon})$ queries.
\end{corollary}

In the $RW^+$ model we can simulate moving knife protocols without the requirement that the valuations are bounded from below since the center can reduce (half) the time directly with each iteration, instead of reducing it through the lens of the players' valuations.

\begin{theorem} \label{thm:simulation_robust_rw+}
	Consider a cake cutting problem where the value densities are bounded from above by constant $D>0$.
	Let $\mathcal{M}$ be an $RW^+$ moving knife protocol with at most $r$ steps, such that $\mathcal{M}$ outputs $\mathcal{F}$-fair allocations demarcated by at most a constant number $C$ of cuts.
	
	Then for each $\epsilon > 0$, there is an $RW^+$ protocol $\mathcal{M}_{\epsilon}$ that uses $O\left(\log \frac{1}{\epsilon}\right)$ queries and computes $\epsilon$-$\mathcal{F}$-fair partitions demarcated with at most $C$ cuts.
\end{theorem}

\medskip
We also introduce a weaker model, which we call $RW^{-}$, where the protocol can ask the players only the evaluate type of query.

\begin{definition}[$RW^-$ query model]
	An $RW^{-}$ protocol for cake cutting communicates with the players via queries of the form
	\begin{itemize}
		\item $\emph{\textbf{Eval}}_i(y)$: Player $i$ returns $V_{i}([0,y])$, where $y \in [0,1]$ is arbitrarily chosen by the center.
	\end{itemize}
	At the end an $RW^-$ protocol outputs an allocation that can be demarcated by any points.
\end{definition}

If the valuations are arbitrary, then an $RW^-$ protocol may be unable to find any fair allocation at all. The reason is that no matter what queries an $RW^-$ protocol asks, one can hide the entire instance in a small interval that has value $1$ for all the players; this interval will shrink as more queries are issued, but can be set to remain of non-zero length until the end of execution. 

\medskip

However, if the valuations are bounded from above\footnote{There exist other types of valuations on which the $RW^-$ model may be useful, such as piecewise constant valuations defined on a grid, with the demarcations between intervals of different height known to the protocol.}, then an $RW^-$ protocol is quite powerful.

\begin{theorem}
	Suppose the valuations of the players are bounded from above by a constant $D > 0$. Then any $RW^+$ query can be answered within $\epsilon$-error using $O(\log{\frac{1}{\epsilon}})$ $RW^-$ queries. 
\end{theorem}
\begin{proof}
	Let there be an instance with arbitrary valuations $v_1 \ldots v_n$ such that $v_i(x) < D$ for all $x \in [0,1]$ and $i \in N$.
	Since an $RW^-$ protocol can use the same type of evaluate queries as an $RW^+$ protocol, the simulation has to handle the case where the incoming query is a cut. Let this be $Cut_i(\alpha)$ and denote by $x$ the correct answer to the query. In order to find an approximate answer using only evaluate queries, initialize $\ell = 0$, $r=1$, and search for the correct answer: $(*)$
	Let $m = (\ell + r)/2$.
	Ask player $i$ the query $Eval_i(m)$ and let $w$ be the answer given. If $|w - \alpha| \leq \epsilon$, return $m$. 
	Otherwise, if $m > \alpha$, set $r = m$, and if $m < \alpha$, set $\ell = m$; return to $(*)$. This procedure halves the interval $[\ell,r]$ with every iteration. Moreover, from the bound on the valuations, an interval of length $\epsilon/D$ cannot be worth more than $\epsilon$ to any player. Thus the search stops in $O(\log\frac{1}{\epsilon})$ rounds.
\end{proof}

\section{Discussion}

An important open question is to obtain stronger lower bounds for $n\geq 4$ players for computing connected envy-free allocations and for $n\geq 3$ players for perfect allocations with minimal number of cuts. We conjecture that unlike equitability, which remains logarithmic in $1/\epsilon$ for any number of players, computing a connected $\epsilon$-envy-free allocation for $n=4$ players and an $\epsilon$-perfect allocation with minimal cuts for $n=3$ players will require $\Omega(\frac{1}{\epsilon})$ queries. Since moving knife protocols can be simulated with $O(\log{\frac{1}{\epsilon}})$ queries, this would imply that no moving knife protocol exists for computing an envy-free allocation for $n \geq 4$ players or a perfect allocation for $n\geq 3$ players (the existence of such procedures has been posed as an open question, e.g. in \cite{BTZ95}).

\section{Acknowledgements}
We thank the participants of the workshop on fair division in St Petersburg for useful discussion, in particular Erel Segal-Halevi and Francis Su.

\addcontentsline{toc}{section}{\protect\numberline{}References}%
\bibliographystyle{alpha}
\bibliography{cake_q}

\appendix

\section{Simulation of Partitions} \label{app:sim_partition}
	
A general technique useful for computing approximately fair allocations in the RW model is based on asking the players to submit a discretization of the cake in many small cells via Cut queries, and reassembling them offline in a way that satisfies approximately the desired solution. The next theorem statement implies that it is possible to approximate in the RW model allocations that exist with a bounded number of cuts, which implies upper bounds for computing approximately fair solutions for several fairness concepts.

\medskip

\noindent \textbf{Lemma \ref{lem:simgen}} (restated)
\emph{Consider a cake cutting problem for $n$ players.
	Suppose there exists a partition $A = (A_1, \ldots, A_m)$, where each piece $A_j$ is demarcated with at most $K$ cuts \footnote{A connected piece $X =[a,b]$ is demarcated by two cuts (namely its endpoints $a$, $b$), a piece $X' = [a,b] \cup [c,d]$ where $a < b < c < d$ is demarcated with four cuts ($a,b,c,d$), and so on. We are interested in the minimum number of points that can be used to demarcate a piece.} and worth $w_{i,j}$ to each player $i$. Then for all $\epsilon > 0$, a partition 
	$\tilde{A} = (\tilde{A}_1, \ldots, \tilde{A}_m)$ where each piece $A_j$ is demarcated with at most $K$ cuts and $|V_i(\tilde{A}_j) - w_{i,j}| \leq \epsilon$ for all $i,j$ can be found with $O(K \cdot n/\epsilon)$ queries.}
\begin{proof}
Let $0 < x_1 < \ldots < x_T < 1$ be the minimum set of points required to demarcate the partition $A = (A_1, \ldots, A_m)$ and consider the RW protocol:
		\begin{itemize}
			\item Ask each player $i$ a number of $B = \lceil 4K/\epsilon \rceil$ Cut queries to partition the cake in $B$ intervals each worth $1/B$ to $i$.\footnote{The queries are: $Cut_i(\epsilon), Cut_i(2\epsilon),\ldots, Cut_i((B-1)\epsilon)$.}
			\item For each partition $\tilde{A} = (\tilde{A}_1, \ldots, \tilde{A}_m)$ that can be attained by assembling offline the resulting $n \cdot B$ cells in $m$ pieces, each demarcated by at most $K$ cut points: 
			\begin{description}
				\item[$\hspace{4mm}\diamond$] Let $u_{i,j}$ be the value of player $i$ for the piece $A_j'$ obtained by rounding the cuts demarcating $\tilde{A}_j$ to player $i$'s own cut points. 
				If $u_{i,j} \in \left[w_{i,j} - \epsilon/2, w_{i,j} + \epsilon/2\right]$ for all $i,j$, then output $\tilde{A}$ and exit.
			\end{description}
		\end{itemize}
		To see that existence of such a partition $\tilde{A}$ is guaranteed, note that by rounding the points $x_1 \ldots x_T$ to the nearest point on the grid submitted by the players and allocating from left to right the resulting intervals in the same order as in $A$ (including empty pieces if rounded cuts overlap), we obtain a partition $\tilde{A}$ in which every piece is demarcated with at most $K$ cuts. The rounding error from each point is $1/B$, thus the overall difference for any piece $A_j$ from the point of view of any player $i$ is bounded by $2K /B \leq \epsilon/2$, so $V_i(\tilde{A}_j) \in \left[w_{i,j} - \epsilon, w_{i,j} + \epsilon\right]$ as required.
	\end{proof}
	
	As a corollary, we get bounds for computing approximately fair allocations in the RW model.  \\
	
		\noindent \textbf{Corollary \ref{cor:apx_general}} (restated)
		\emph{For each $\epsilon > 0$ and number of players $n$ with arbitrary value densities, a partition that is}
		\begin{itemize}
				\item \emph{$\epsilon$-envy-free and connected can be computed with $O(n/\epsilon)$ queries}
			\item \emph{$\epsilon$-equitable, proportional, and connected can be computed with $O(n/\epsilon)$ queries}
			\item \emph{an $(\epsilon,k)$-measure splitting with $k$ pieces can be computed with $O(k\cdot n^2/\epsilon)$ queries.}
			\item \emph{$\epsilon$-perfect can be computed with $O(n^3/\epsilon)$ queries.}
		\end{itemize}

\section{Envy-Free Allocations} \label{app:ef}

In this section we include the proofs for computing lower and upper bounds for the computation of envy-free allocations.

\subsection{Upper Bound}
To obtain a logarithmic upper bound for computing connected $\epsilon$-envy-free allocations, we will simulate from the point of view of each player a moving knife procedure due to \cite{BB04}.

\begin{quote}
	\textbf{\emph{Barbanel-Brams procedure}}:
	\emph{Ask each player $i$ to return the point such that one third of the cake is to the right of it.} 
	\emph{If an envy-free allocation can be formed with the pieces demarcated by the player who had the rightmost mark (say 1)---$\{[0,\ell], [\ell,r],[r,1]\}$---output it. Otherwise there are two cases:} 
	
	\hspace{6mm} \emph{Case 1: Both players $2$ and $3$ strictly prefer the piece $[\ell, r]$. Then move a sword continuously from $r$ towards $\ell$, keeping for each position $z$ of the sword, the point $t$ for which $V_1(0,t) = V_1(z,1)$. By the intermediate value theorem, there exists position of the sword such that one of $2$, $3$ is indifferent between two pieces, and an envy-free allocation exists at cuts $t$,$z$.}
	
	\hspace{6mm} \emph{Case 2: Both players $2$ and $3$ strictly prefer the piece $[0, \ell]$. Then move a sword continuously from $\ell$ towards $0$, keeping for each position $z$ of the sword the point $t$ such that $V_1(z,t) = V_1(t,1)$. By the intermediate value theorem, there exists position of the sword such that one of the players $2$ and $3$ is indifferent between the piece $[0,z]$ and one of $\{[z,t], [t,1]\}$, which yields an envy-free allocation with cuts $t,z$.}
\end{quote}

\noindent \textbf{Theorem \ref{thm:ef3sim}} (restated)
\emph{A connected $\epsilon$-envy-free allocation for three players can be computed with $O\left(\log{\frac{1}{\epsilon}}\right)$ queries.}\\
\begin{proof} 
	We will compute an $\epsilon$-envy-free allocation using several steps, not all of which are not included in the Barbanel-Brams protocol. The first step of the Barbanel-Brams protocol is discrete, and so it can be executed with $O(1)$ RW queries. 
	
	If the instance falls in Case 1, we will maintaining the next property: 
	\begin{enumerate}[\hspace{4mm}(a)]	
		\item there exist points $0 < w < z < 1$, such that there for some points $a < b$, $a,b \in [0,w)$ we have $V_1(0,a) = V_1(z,1)$, $V_1(0,b) = V_1(w,1)$, and players $2$ and $3$ agree that $(i)$ the piece $[a,z]$ is larger than any of $[0,a]$ and $[z,1]$ by more than $\epsilon$, while $(ii)$ the piece $[b,w]$ is smaller than one of $[0,b]$ and $[w,1]$ by more than $\epsilon$.
	\end{enumerate}
	
	Consider the following procedure:
	\begin{enumerate}
		\item Initialize $w$ and $z$ by setting $z = r$ and asking player $1$ a Cut query to identify its midpoint $w$ of the cake. Clearly the allocation made with cut points $(\ell,r)$ has the property that both players $2$ and $3$ prefer the middle piece, $[\ell,r]$, while if the cuts overlap on $w$, then both $2$ and $3$ prefer one of the pieces $[0,w]$ or $[w,1]$.
		\begin{figure}[h!]
			\centering
			\includegraphics[scale=0.7]{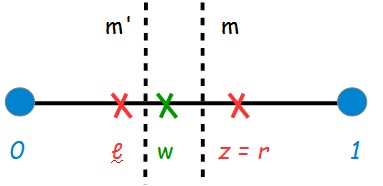}
			\caption{Case 1 of the simulation.}
			\label{fig:bbcase1}
		\end{figure}
		
		\item Given points $w,z$ that satisfy invariant $(a)$, for which $V_1(w,z) \geq \epsilon$, ask player $1$ iteratively a Cut query to determine the midpoint $m \in [w,z]$, for which $V_1(w,m) = V_1(m,z)$. Then ask player $1$ a Cut query as to return the point $m'$ for which $V_1(0,m') = V_1(m,1)$. If there is an $\epsilon$-envy-free allocation with cuts $m'$ and $m$, output it. Otherwise, if both players $2$ and $3$ evaluate the piece $[m',m]$ as the largest among $\{[0,m'], [m',m], [m,1]\}$, then set $z = m$. Else, it must be the case that both $2$ and $3$ estimate the piece $[m',m]$ as strictly smaller than at least one of $[0,m']$ and $[m,1]$ by more than $\epsilon$; set $w = m$. Note the new points $w,z$ still satisfy property $(a)$.
	\end{enumerate}
	Step 1 requires a constant number of queries, while Step 2 is executed at most $\log{\frac{1}{\epsilon}}$ times, since the valuation of player $1$ for the interval $[w,z]$ halves with each round; each round of Step 2 requires a constant number of queries. If an $\epsilon$-envy-free allocation has been found after completing Steps 1 and 2, the case is complete. Otherwise, we will further reduce the interval $[w,z]$ until it becomes small also from the point of view of player $2$. Note that 
	the invariant still holds after completing the previous steps, so there exist $a < b < w$ such that 
	$V_1(w,z) = V_1(a,b) = \delta < \epsilon$. By the intermediate value theorem, we are guaranteed to have an $\epsilon$ envy-free solution with cuts $x_1 \in [a,b]$ and $x_2 \in [w,z]$. Moreover,
	since $V_1([a,z]) \leq 1/3$, any allocation obtained with cuts $x_1 \in [a,b]$, $x_2 \in [w,z]$ such that player $1$ receives one of the pieces $[0,x_1]$, $[x_2, 1]$ is $\epsilon$-envy-free for player $1$. 
	
	Recall the piece $[a,z]$ is strictly larger than $\epsilon$ than $[0,a]$ and $[z,1]$ in the estimation of players $2$ and $3$. We have three subcases, depending on the piece $[b,z]$. 
	\begin{itemize}
		\item Exactly one of players $2$ and $3$ views $[b,z]$ as the largest piece among $[0,b]$, $[b,z]$, and $[z,1]$. Then an $\epsilon$-envy-free allocation can be obtained with cuts $b$ and $z$.
		\item Both players $2$ and $3$ view $[b,z]$ as larger than $[0,b]$ and $[z,1]$. Then there must exist an $\epsilon$-envy-free solution with cuts $b$ and $x_2 \in [w,z]$. Such a solution can be found with binary search on $[w,z]$, using the valuation of player $2$ to half the interval $[w,z]$ in each iteration. The solution $x_2$ reached this way will have the property that player $2$ is indifferent (within $\epsilon$) between $[b,x_2]$ and one of the outside pieces $[0,b]$ or $[x_2,1]$. 
		\item Otherwise, both players $2$ and $3$ view the piece $[b,z]$ as smaller than one of $[0,b]$ and $[z,1]$. Then by the intermediate value theorem there exists an envy-free allocation with cuts $x_1 \in [a,b]$ and $z$. We can find an approximate solution with binary search on the interval $[a,b]$ using the valuation of player $2$ to repeatedly identify the midpoint of $[a,b]$.
	\end{itemize}  
	
	Each subcase completes with $O\left(\log{\frac{1}{\epsilon}}\right)$ queries, which gives a bound of $O\left(\log{\frac{1}{\epsilon}}\right)$ for Case 1.
	Otherwise, we enter Case 2, for which we maintain the invariant:
	\begin{enumerate}[\hspace{4mm}(b)]	
		\item there exist points $0 < w < z < 1$, such that for some points $a,b$ with $w < a < b \leq 1$ we have $V_1(w,a) = V_1(a,1)$, $V_1(z,b) = V_1(b,z)$, and players $2$ and $3$ agree that $(i)$ the piece $[0,z]$ is larger than any of $[z,b]$ and $[b,1]$ by more than $\epsilon$, while $(ii)$ the piece $[0,w]$ is smaller than one of $[w,a]$ and $[a,1]$ by more than $\epsilon$.
	\end{enumerate}
	The steps for simulating Case 2 are:
	
	\begin{enumerate}
		\item[3.] Initialize $w=0$ and $z = \ell$. Ask player $1$ a Cut query to identify the midpoint $q$ of the cake in its estimation. An allocation made with cuts $\ell$ and $r$ has the property that both players $2$ and $3$ view the leftmost piece $[0,\ell]$ as the largest, while the allocation made with pieces $\{[0,0], [0,q], [q,1]\}$ has the property that none of the players $2$ and $3$ want the (now empty) leftmost piece.
		\item[4.] Given points $w,z$ that satisfy invariant $a$, for which $V_1(w,z) \geq \epsilon$, iteratively ask player $1$ a Cut query to determine point $m \in [w,z]$ for which $V_1(w,m) = V_1(m,z)$. Then find, via another Cut query, the point $m' \in [z,1]$ for which $V_1(z,m') = V_1(m',1)$. If an $\epsilon$-envy-free allocation exists with cuts $m$ and $m'$, output it. Otherwise, if both players $2$ and $3$ strictly prefer piece $[0,m]$ to any of $[m,m']$ and $[m',1]$, then update $z = m$. Else, both $2$ and $3$ view $[0,m]$ as strictly smaller than at least one of $[m,m']$ and $[m',1]$ by more than $\epsilon$; set $w = m$.
		\begin{figure}[h!]
			\centering
			\includegraphics[scale=0.7]{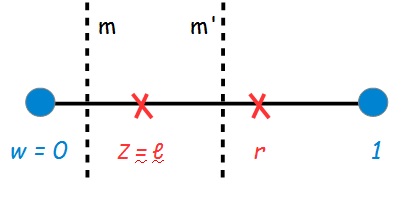}
			\caption{Case 2 of the simulation.}
			\label{fig:bbcase2}
		\end{figure}
	\end{enumerate}
	Step 3 requires a constant number of queries, while Step 4 at most $\log\frac{1}{\epsilon}$ queries. If an $\epsilon$-envy-free allocation has not been found after completing steps $3$-$4$, then since the interval $[w,z]$ is worth less than $\epsilon$ to player $1$, we can again reduce the problem to finding an agreement among players $2$ and $3$, as was done in Case 1. This completes the proof.
\end{proof}

\subsection{Lower Bound}
 
In this section we develop the lemmas for the lower bound. 

\medskip

\noindent \textbf{Theorem \ref{thm:3lb}} (restated).
\emph{Computing a connected $\epsilon$-envy-free allocation for three players requires $\Omega\left(\log \frac{1}{\epsilon}\right)$ queries.}\\

Recall the class of valuations that will be used for the lower bound is that of generalized rigid measure systems.

\medskip

\noindent \textbf{Definition \ref{def:rms}} (restated).
\emph{A tuple of value densities $(v_1, \ldots, v_n)$ is a generalized rigid measure system if:
	\begin{itemize}
		\item the density of each measure is bounded by: $\frac{1}{\sqrt{2}} < v_i(x) < \sqrt{2}$, for all $i =1 \ldots n$ and $x \in [0,1]$.
		\item there exist points $0 = x_0 < x_1< \ldots< x_{n-1}< x_n =1$ and values $s_i,t_i$ for each player $i$ such that $0 < s_i < 1/n < t_i < 1/2$ and the valuations of the players satisfy the constraints: $V_j(x_{j-1}, x_j) = V_j(x_j, x_{j+1}) = t_j$ for all $j = 1 \ldots n-1$ and $V_n(x_{n-1}, x_n)= V_n(0,x_1)=t_n$. 
	\end{itemize}
}
An example for three players is given next.

\begin{table}[h!]
	\label{tab:grms}
	\begin{center}
		\begin{tabular}{l | c | c | r}
			
			& $[0,x]$ & $[x,y]$ & $[y,1]$ \\ \hline
			$V_1$ & $t_1$ & $t_1$ & $s_1$ \\
			$V_2$ & $s_2$ & $t_2$ & $t_2$ \\
			$V_3$ & $t_3$ & $s_3$ & $t_3$ \\
		\end{tabular}
	\end{center}
\end{table}

Generalized rigid measure systems satisfy the property that the valuations of the players for any given piece cannot differ too much. 

\medskip

\noindent \textbf{Lemma \ref{lem:density}} (restated)
\emph{Consider any cake cutting problem where for two players $i$ and $j$ 
	where there exist $a,b>0$ such that for all $x \in [0,1]$, $1/a < v_i(x) < b$ and $1/a < v_j(x)  < b$.
	Then for any two pieces $S_1,S_2$ of the cake, if
	$V_i(S_1) \geq ab \cdot V_i(S_2)$, it follows that $ab \cdot V_j(S_1) > V_j(S_2)$.}
\begin{proof}
	Let $S_1$ and $S_2$ be two pieces of lengths $\ell_1$ and $\ell_2$, respectively, such that $V_i(S_1) \geq ab \cdot V_i(S_2)$. By using the constraints on the densities, we get:
	$
	ab \cdot V_j(S_1) > ab \cdot \left(\frac{\ell_1}{a} \right) = b \cdot \ell_1 > V_i(S_1) \geq ab \cdot V_i(S_2) > ab \cdot \left( \frac{\ell_2}{a} \right)= b \cdot \ell_2 > V_j(S_2)
	$ as needed.
\end{proof}

A useful notion to measure how close a protocol is to discovering an approximately envy-free solution on a given instance will be that of a partial rigid measure system, which we define for three players; the $n>3$ definition works similarly.

\medskip

\noindent \textbf{Definition \ref{def:prms}} (restated)
\emph{A tuple of value densities $(v_1, v_2, v_3)$ is a partial rigid measure system if}
	\begin{itemize}
		\item \emph{the density of each player $i$ is bounded everywhere: $\frac{1}{\sqrt{2}} < v_i(x) < \sqrt{2}$, for all $x \in [0,1]$.}
		\item \emph{there exist values $k > 0$ and $1/2 > \ell_i > 1/3 > m_i > 0$ for each player $i$, and points $x,x',y,y' \in [0,1]$, so that the matrix of valuations for pieces demarcated by these points is:}
		\begin{table}[h!]
			\label{tab:prms}
			\begin{center}
				\begin{tabular}{l | c | c | c | c | c r}
					
					& $[0,x]$ & $[x,x']$ & $[x',y]$ & $[y,y']$ & $[y',1]$ \\ \hline
					$V_1$ & $\ell_1$ & $k$ & $\ell_1$ & $k$ & $m_1$ \\
					$V_2$ & $m_2$ & $k$ & $\ell_2$ & $k$ & $\ell_2$ \\
					$V_3$ & $\ell_3$ & $k$ & $m_3$ & $k$ & $\ell_3$ 
				\end{tabular}
			\end{center}
		\end{table}
	\end{itemize}
\vspace{-8mm}

Partial rigid measure systems have the property that if there is a collection of cut points $\mathcal{P} = \{z_1, \ldots, z_k\} \subset [0,1]$, such that there are no points from $\mathcal{P}$ in the intervals $(x,x')$ and $(y,y')$, then any connected partition attainable using cut points from $\mathcal{P}$ has envy at least $0.01 k$.

\medskip

\noindent \textbf{Lemma \ref{lem:far_ef3}} (restated)
\emph{Let $(v_1, v_2, v_3)$ be a partial rigid measure system with parameters $k, m_i, \ell_i$ so that for each $i$, $V_i([x,x']) = V_i([y, y']) = k$ for  some $x,x',y,y'$. 
	If $\mathcal{P} = \{z_1, \ldots, z_{t}\} \subset [0,1]$ is a collection of cut points such that $x,x',y,y' \in \mathcal{P}$ and there are no points from $\mathcal{P}$ in the intervals $(x,x')$ and $(y,y')$, then any connected partition demarcated by points in $\mathcal{P}$ has envy at least $0.01 k$.}
\begin{proof}
Suppose towards a contradiction there exists an allocation with cut points $\hat{x},\hat{y} \in \mathcal{P}$ so that each player is not envious by more than $0.01k$. Recall that $\hat{x}, \hat{y} \not\in (x,x') \cup (y,y')$. We show that every assignment leads to high envy or to a contradiction. There are a few cases: 
\medskip	

\noindent \emph{Case 1}: $\hat{x} \leq x$. 
Regardless of who owns the piece $[0,\hat{x}]$, it cannot be the case that $\hat{y} \leq x$, since that player would envy the right remainder of the cake, namely the interval $[\hat{y},1]$, by an amount much larger than $k$. We consider several scenarios, depending on the owner of $[0, \hat{x}]$. 

\begin{description}
\item[]\emph{1.a)} Player $1$ receives $[0,\hat{x}]$. There are two subcases: 
\begin{itemize}
\item $\hat{y} \in [x', y]$. 
Let $S_1 = [\hat{x}, x]$ and $S_2 = [\hat{y}, y]$. 
Since player $1$ does not envy any other piece by more than $0.01k$, we have 
\begin{eqnarray*}
V_1(0,\hat{x}) & \geq & V_1(\hat{x}, \hat{y}) - 0.01 k \iff \ell_1 - V_1(S_1) = V_1(0,x) - V_1(S_1) \\
& \geq & V_1(x,y) + V_1(S_1) - V_1(S_2) - 0.01 k = \ell_1 + k + V_1(S_1) - V_1(S_2) - 0.01 k  \\
&\iff & V_1(S_2) \geq  2\left(V_1(S_1) + 0.495 k\right)
\end{eqnarray*}

By Lemma \ref{lem:density}, we have that $2 V_j(S_2) > V_j(S_1) + 0.495 k$ for all $j \in \{2,3\}$.
Then both players $2$ and $3$ will prefer the rightmost piece, $[\hat{y}, 1]$, by a margin of at least $0.01 k$. For player $2$, we have
\begin{eqnarray*}
V_2(\hat{y},1) - V_2(\hat{x}, \hat{y}) & = & V_2(\hat{y}, y) + k + V_2(y',1) - (V_2(x,y) + V_2(S_1) - V_2(S_2)) \\
& = & k + V_2(S_2) + \ell_2 - 
(\ell_2 + k + V_2(S_1) - V_2(S_2)) = 2 V_2(S_2) - V_2(S_1) \\
& > & 0.495 k > 0.01k
\end{eqnarray*}
For player $3$, we have	
\begin{eqnarray*}
V_3(\hat{y},1) - V_3(\hat{x}, \hat{y}) & = & V_3(\hat{y}, y) + k + V_3(y',1) - (V_3(x,y) + V_3(S_1) - V_3(S_2)) \\
& = & 2V_3(S_2) - V_3(S_1) + (\ell_3 - m_3) > 0.495 k + \ell_3 - m_3 > 0.01 k
\end{eqnarray*}
\item $\hat{y} \geq y'$.
Then player $1$ will envy the middle piece by at least $2k$:
\begin{eqnarray*}
V_1(\hat{x}, \hat{y}) - V_1(0,\hat{x}) = V_1(\hat{x},x) + k + \ell_1 + k + V_1(y',\hat{y}) - V_1(0,\hat{x}) \geq \ell_1 + 2k - \ell_1 = 2k
\end{eqnarray*}
\end{itemize} 

\item[]\emph{1.b)} Player $2$ receives $[0,\hat{x}]$. Then by approximate envy-freeness, we have 
$$V_2(0,\hat{x}) \geq V_2(\hat{x},\hat{y}) - 0.01 k \; \; \mbox{and} \; \;
V_2(0,\hat{x}) \geq V_2(\hat{y}, 1) - 0.01 k
$$
By summing up the two inequalities, we get 
\begin{eqnarray*}
V_2(0,\hat{x}) & \geq & V_2(\hat{x}, 1)/2 - 0.01 k = (1 - V_2(0,\hat{x}))/2 - 0.01 k \iff 3 V_2(0,\hat{x})/2 \geq 1/2 - 0.01 k \\
& \iff & V_2(0,\hat{x}) \geq 1/3 - 0.02 k /3
\end{eqnarray*}

Since $m_2 \geq V_2(0,\hat{x})$, $2 \ell_2 + m_2 + 2k = 1$, and $\ell_2 > 1/3$, we get $m_2 = 1 - 2 \ell_2 - 2k < 1/3 - 2k$, so
$
1/3 - 2k > m_2 \geq V_2(0,\hat{x}) \geq 1/3 - 0.02k/3$.
Contradiction, so the case cannot happen. 

\item[]\emph{1.c)}	Player $3$ receives $[0,\hat{x}]$. We have the following subcases:
\begin{itemize}
\item $\hat{y} \in [x', y]$. Then $V_3(0,\hat{x}) \leq \ell_3$, while $V_3(\hat{y},1) \geq \ell_3 + k$, thus player $3$ envies the piece $[\hat{y},1]$ by more than $0.01 k$.
\item $\hat{y} \geq y'$. Then both players $1$ and $2$ value the middle piece $[\hat{x},\hat{y}]$ more than the rightmost piece $[\hat{y},1]$, and the difference is larger than $0.01 k$. For player $2$ we have 
$V_2(\hat{x},\hat{y}) \geq \ell_2 + 2k$, while $V_2(\hat{y},1) \leq \ell_2$, and for player $1$ we have $V_1(\hat{x},\hat{y}) \geq \ell_1 + 2k$, while $V_1(\hat{y},1) \leq m_1 < \ell_1$.
\end{itemize} 

Thus no player can accept the leftmost piece, $[0,\hat{x}]$, which completes Case 1.
\end{description}

\noindent \emph{Case 2}: $\hat{x} \in [x',y]$. Then $\hat{y} \geq y'$, since no player would accept (even within envy $0.01 k$) a piece smaller than $[x',y]$. Moreover, players $1$ and $3$ would not accept the piece $[\hat{y}, 1]$ since they would envy the player owning $[0,\hat{x}]$ by more than $k$. Thus the piece $[\hat{y},1]$ can only be assigned to player $2$.
Let $S_1 = [x',\hat{x}]$ and $S_2 = [y',\hat{y}]$. By approximate envy-freeness of player $2$, we have 
\begin{eqnarray*}
V_2(\hat{y},1) \geq V_2(\hat{x},\hat{y}) - 0.01 k = V_2(x',y) - V_2(S_1) + V_2(S_2) + k - 0.01 k = \ell_2 - V_2(S_1) + V_2(S_2) + 0.99k
\end{eqnarray*}
Since $V_2(\hat{y},1) = \ell_2 - V_2(S_2)$, we get that 
$$
\ell_2 - V_2(S_2) \geq \ell_2 - V_2(S_1) + V_2(S_2) + 0.99k \iff V_2(S_1) \geq 2 \left(V_2(S_2) + 0.495 k \right)
$$
By Lemma \ref{lem:density}, we have $2 V_j(S_1) \geq V_j(S_2) + 0.99k$ for all $j \in \{1,3\}$.

We wish to show that both players $1$ and $3$ would only accept the leftmost piece $[0,\hat{x}]$ from the remaining pieces. For player $1$, we have 
\begin{eqnarray*}
V_1(0,\hat{x}) - V_1(\hat{x},\hat{y}) & = & V_1(0,x) + V_1(x,x') + V_1(S_1) - (V_1(x',y) - V_1(S_1) + k + V_1(S_2)) \\
& = & \ell_1 + k + V_1(S_1) - \ell_1 + V_1(S_1) - k - V_1(S_2) = 2 V_1(S_1) - V_1(S_2)\\
& \geq & 0.495k > 0.01 k
\end{eqnarray*}
For player $3$, using the fact that $\ell_3 > m_3$, we have 
\begin{eqnarray*}
V_3(0,\hat{x}) - V_3(\hat{x},\hat{y}) & = & V_3(0,x) + V_3(x,x') + V_3(S_1) - (V_3(x',y) - V_3(S_1) + k + V_3(S_2)) \\
& = & \ell_3 + k + V_3(S_1) - m_3 + V_3(S_1) - k - V_3(S_2) = 2 V_3(S_1) - V_3(S_2) + \ell_3 - m_3\\
& \geq & 0.495 k + \ell_3 - m_3 > 0.01 k
\end{eqnarray*}

\noindent \emph{Case 3}: $\hat{x} \geq y'$. This scenario is clearly infeasible, since all the players would envy the owner of the piece $[0,\hat{x}]$ by at least $2k$. 

In all the cases we obtained a contradiction, so any partition attained with cut points from $\mathcal{P}$ has envy at least $0.01 k$. This completes the proof.
\end{proof}

\medskip

Next we show how queries can be answered one at a time so that the valuations remain consistent with (some) partial rigid measure system throughout the execution of a protocol. \\

\noindent \textbf{Lemma \ref{lem:hide}} (restated)
\emph{Suppose that at some point during the execution of an $RW$ protocol for three players the valuations and cuts discovered are consistent with a partial rigid measure system with parameters $k > 0$, $0 < m_i < 1/3 < \ell_i < 1/2$ for each $i$, and points $x,y$, so that the valuations are:}
	\begin{figure}[h!] 
		\centering
		\includegraphics[scale=0.85]{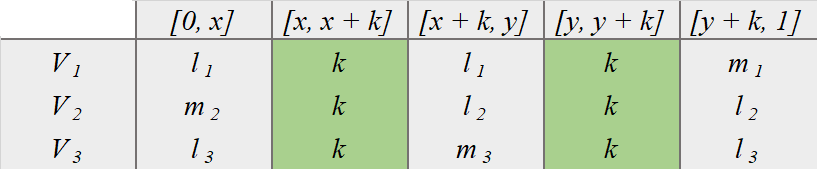}
		\label{fig:partial_main}
	\end{figure}	

\emph{If the intervals $I = [x,x+k]$ and $J = [y,y+k]$ have no cut points inside, then a new cut query can be answered so that the valuations remain
	consistent with a partial rigid measure system where two new intervals $I' \subseteq I$ and $J' \subseteq J$ have no cuts inside, length $0.01k$, and the densities of all the players are uniform on $I',J'$.}\\

We will show how to set the values when the query is addressed to player $1$. The other cases, when players $2$ and $3$ have to answer, are similar, but included for completeness.

\bigskip

\noindent \textbf{Part a of Lemma \ref{lem:hide}} 

\begin{proof}
We show how to set the values when the query is addressed to player $1$: $Cut_1(\alpha)$. If the answer to the cut query falls outside $I$ and $J$, we can answer in a way that is arbitrary but consistent with the history. The more difficult scenario is when the answer must be a point inside $I$ or $J$, for which we must hide the solution inside smaller, but not too small, new intervals $I' \subseteq I$ and $J' \subseteq J$. 	Let $\mathcal{P}$ be the collection of cut points that have been discovered by $\mathcal{A}$ before the query is received. By symmetry, it will be sufficient to solve the problem where the new cut query falls in interval $I$; thus $\alpha \in (\ell_1, \ell_1 +k)$.
If the cut query falls on the left side of interval $I$, we hide the solution in a subinterval $I'$ (of length 100 times smaller than $I$) on the right side of $I$, and viceversa. \\

\noindent \emph{Case 1}: $\alpha \in \left(\ell_1, \ell_1 + 2k/3\right]$. We will hide the new interval $I'$ on the right side of the point $x + 2k/3$. Let 
$m = x + 2k/3$, $n = x + 2k/3 + 0.01k$, and $I' = [m,n]$. Similarly let
$p = y + 1.03k/3$, $q = y + 1.03k/3 + 0.01k$, and $J' = [p,q]$. Update the collection of cut points to $\mathcal{P}' = \mathcal{P} \cup \{m,n,p,q\}$. Let the value of each player $i$ for the intervals $I'$ and $J'$ be exactly $0.01k$. Since $I'$ and $J'$ have length $0.01k$, the densities remain uniform in these intervals. Set the values of player $1$ for the other new intervals to 
$V_1(x,m) = 2k/3, \; V_1(n,x+k) = 0.97k/3, \; V_1(y,p) = 1.03k/3$, and $V_1(q,y+k) = 1.94k/3.$

Denote the values of player $2$ for the unknown intervals by 
$V_2(x,m) = 0.99k - \lambda k, \; V_2(n,x+k) = \lambda k, \; V_2(y,p) = \mu k$, $V_2(q,y+k) = 0.99k - \mu k,$ where $0 < \lambda, \mu < 0.99$. Note the values add up to the weight of $I$ and $J$ for player $2$:
\begin{itemize} 
\item $V_2(x,x+k) = V_2(x,m) + V_2(m,n) + V_2(n,x+k) = 0.99k - \lambda k + 0.01k + \lambda k = k$
\item $V_2(y,y+k) = V_2(y,p) + V_2(p,q)+ V_2(q,y+k) = \mu k + 0.01k + 0.99k - \mu k = k
$
\end{itemize}
For player $3$, the values of the unknown intervals are
$
V_3(x,m) = 0.99k - \phi k, \; V_3(n,x+k) = \phi k, \; V_3(y,p) = \psi k, \; \mbox{and} \; V_3(q,y+k) = 0.99k - \psi k,
$
where $0 < \phi, \psi < 0.99k$. The weights add up to the value of player $3$ for the intervals $I$ and $J$ through a check similar to the one for player $2$.

\bigskip

We obtain the configuration in Figure \ref{fig:partial_1a}, where the parameters $\lambda, \mu, \phi, \psi$ must be determined.

\begin{figure}[h!]
\centering
\includegraphics[scale=0.85]{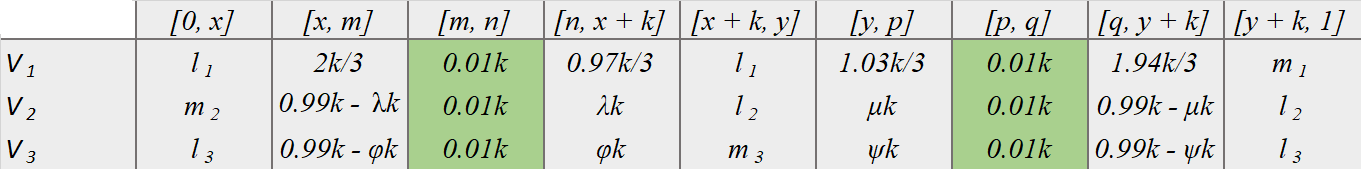}
\caption{Partial rigid measure system for $\alpha \in \left(\ell_1, \ell_1 + 2k/3\right]$. The break points are $m = x + 2k/3$, $n = x + 2k/3 + 0.01k$, $p = y + 1.03k/3$, $q = y + 1.03k/3 + 0.01k$. The densities are uniform on $[m,n]$, $[p,q]$.}
\label{fig:partial_1a}
\end{figure}
The remaining requirements for the new configuration to form a partial rigid measure system are that the values yield a new configuration with parameters $\ell_i', m_i'$ and all the densities on $[x,m]$, $[n,x+k]$, $[y,p]$, and $[q,y+k]$ are in the required bounds of $1/\sqrt{2}$ and $\sqrt{2}$. 

For player $1$, by choice of values we have that $\ell_1' = \ell_1 + 2k/3 = 0.97k/3 + \ell_1 + 1.03k/3$ and $m_1' = m_1 + 1.94k/3$.
Moreover, it can be verified that player $1$'s density will be uniform on all of the new intervals, which clearly belongs to the range
$\left(1/\sqrt{2}, \sqrt{2}\right)$.
For player $2$ we must find $0 < \lambda,\mu < 0.99$ such that the values still form a partial rigid measure system. Thus
$\lambda k + \ell_2 + \mu k = 0.99k - \mu k + \ell_2 \iff \lambda + 2 \mu = 0.99 \iff \mu = 0.495 - \lambda/2.$
By choice of the points $m,n,p,q$ and the valuations, the density of player $2$ on the different intervals is: $(297 - 300 \lambda)/200$ on $[x,m]$, $300 \lambda/97$ on $[n,x+k]$,
$297/206 - 150 \lambda / 103$ on $[y,p]$, and $297/388 + 150 \lambda/194 $ on $[q,y+k]$.

Set $\lambda = 0.25$. Then the density is in $(1/\sqrt{2}, \sqrt{2})$ and there exists a solution for player $2$. Let $\mu = 0.37$, $m_2' = m_2 + 0.99k - 0.25k$, and $\ell_2' = \lambda k + \ell_2 + \mu k = 99k/100 - \mu k + \ell_2$. 

Finally, for player $3$, we must find $0 < \phi, \psi < 0.99$ so that
$
\ell_3 + 0.99k - \phi k = 0.99k - \psi k + \ell_3 \iff
\phi = \psi
$. The density of player $3$ is: 
$1.485 - 1.5 \phi$ on $[x,m]$, $300 \phi/97$ on $[n,x+k]$, $300 \phi/103$ on $[y,p]$, and $297/194 - 300 \phi/194$ on $[q,y+k]$.
By setting $\phi = \psi = 0.25$, we obtain correct range for the density of player $3$. Let $\ell_3' = \ell_3 + 0.99k - 0.25k$ and $m_3' = m_3 + 0.5k$.

Now we can answer the query. Since $\ell_1 < \alpha < \ell_1 + 2k/3$, and the values of all the players on the interval $[x,m]$ have been set, find the point $z$ with the property that $V_1(x,z) = \alpha$ and the density is uniform for player $1$ on $[x,z]$. Then fit the answers for the other two players, proportional with their average density on $[x,z]$, and update $\mathcal{P}'$ to include the point $z$. \\

\noindent \emph{Case 2}: $\alpha \in \left(\ell_1 + 2k/3, \ell_1 + k\right)$. This time the interval $I'$ will be hidden on the left side of $x + 2k/3$.
Define
$
m = x + 2k/3 - 0.01k$ and $n = x + 2k/3$.
Set $I' = [m,n]$. Let 
$p = y + 0.97k/3$ and $q = y + k/3$.
Set $J' = [p,q]$. Update the collection of cut points to $\mathcal{P}' = \mathcal{P} \cup \{m,n,p,q\}$. Let the value of each player $i$ for the intervals $I'$ and $J'$ be exactly $0.01k$; again the densities remain uniform on $I'$ and $J'$. Set the values of player $1$ for the other intervals to 
$V_1(x,m) = 2k/3 - 0.01k, \; V_1(n,x+k) = k/3, \; V_1(y,p) = 0.97k/3$, and $V_1(q,y+k) = 2k/3.$ It can be verified that player $1$'s density is uniform on all the new intervals. Update $\ell_1' = \ell_1 + 2k/3 - 0.01k$ and $m_1' = m_1 + 2k/3$.

For players $2$ and $3$ we must find parameters $0 < \lambda, \mu, \phi, \psi < 0.99$ such that the matrix of valuations in Figure \ref{fig:partial_1b} is consistent with a partial rigid measure system.

\begin{figure}[h!]
\centering
\includegraphics[scale=0.85]{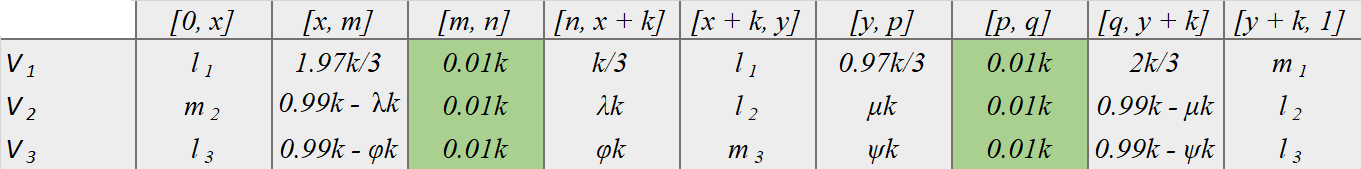}
\caption{Partial rigid measure system for $\alpha \in \left(\ell_1 + 2k/3, \ell_1 + k\right)$. The break points are $m = x + 2k/3 - 0.01k$, $n = x + 2k/3$, $p = y + 0.97k/3$, $q = y + k/3$; uniform densities on $[m,n]$ and $[p,q]$.}
\label{fig:partial_1b}
\end{figure}

\noindent For player $2$ we obtain from Case 1 that $\mu = 99/200 - \lambda/2$. The density of player $2$ is: $297/197 - 300 \lambda / 197$ on $[x,m]$,  $3 \lambda$ on $[n,x+k]$, $300 \mu / 97$ on 
$[y,p]$, and $297/200 - 3\mu/2$ on $[q,y+k]$.
Setting $\lambda = 0.25$ ensures the density on each of these intervals is in $(1/\sqrt{2}, \sqrt{2})$. Then $\mu = 0.37$. Update $\ell_2' = \ell_2 + 0.62k$ and $\mu_2' = m_2 + 0.74k$.

For player $3$, by symmetry with Case 1, we have $\phi = \psi$ and $0 < \phi < 0.99$.
The density of player $3$ is $297/197 - 300 \phi/197$ on $[x,m]$,$ 3 \phi$ on $[n,x+k]$, $300 \phi/103$ on $[y,p]$, and $297/200 - 3\phi/2$ on $[q,y+k]$.
Setting $\phi = \psi = 0.25$ ensures the density on each interval is in the $(1\sqrt{2}, \sqrt{2})$ range. Update $\ell_3' = \ell_3 + 0.74k$ and $m_3' = m_3 + 0.5k$.

We can now answer the query addressed to player $1$. The interval $[0,m]$ has the property that $V_1(0,m) = \ell_1 + 2k/3$, and so $\alpha > V_1(0,m)$. Thus we can return a point $z \in (n, x+k)$ with the property that player $1$'s density is uniform on $[n,z]$. Add $z$ to $\mathcal{P}'$ and report the answers of the other players for the piece $[n,z]$ in a way that is proportional to their average density on $[n,x+k]$. 

Thus if player $1$ receives a query falling inside interval $I$, we can find answers so that the new configuration is still a partial rigid measure system with the properties required by the lemma.
\end{proof}

\medskip

\noindent \textbf{Part b of Lemma \ref{lem:hide}}

\begin{proof} 
	Here the query is for player $2$, say $Cut_2(\alpha)$. We have two cases: \\
	
	\noindent \emph{Case 1}: $\alpha \in \left(m_2, m_2 + 2k/3\right]$. We will hide the interval at the right of the point $x + 2k/3$. Define 
	$m = x + 2k/3$ and $n = x + 2k/3 + 0.01k$.
	Set $I' = [m,n]$. Let 
	$
	p = y + k/3$ and $q = y + k/3 + 0.01 k
	$
	Set $J' = [p,q]$. Let the values of all the players be $0.01k$ in the intervals $I'$ and $J'$. Since the length of these intervals is exactly $0.01 k$, it follows that everyone's densities are uniform herein.
	
	Set the density of player $2$ uniform on each of the generated subintervals. Then $0.97k/3 + \ell_2 + k/3 = 1.97k/3 + \ell_2.$
	Thus for player $2$  the values are consistent with a partial rigid measure system. Set $\ell_2' = \ell_2 + 1.97k/3$ and $m_2' = m_2 + 2k/3$.
	
	We must now fit the valuations of players $1$ and $3$, which implies again finding parameters $0 < \lambda, \mu, \phi, \psi < 0.99$ so that the matrix of valuations is:
		\begin{figure}[h!]
	\centering
	\includegraphics[scale=0.85]{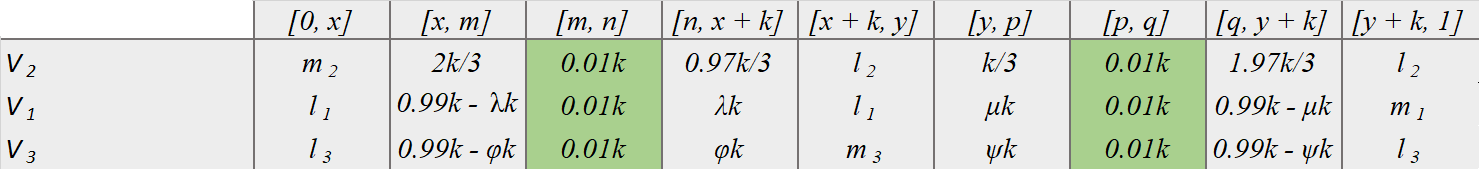}
	\label{fig:partial_2a}
\end{figure}

	For the valuation of player $1$ to be consistent with a partial rigid measure system, 
	$
	\ell_1 + 0.99k - \lambda k = \lambda k + \ell_1 + \mu_k \iff \mu = 0.99 - 2 \lambda
	$.
	The density of player $1$ must lie in $\left(1/\sqrt{2}, \sqrt{2}\right)$ and is given by: $297/200 - 3 \lambda/2$ on the interval $[x,m]$, $300 \lambda/97$ on $[n,x+k]$, $297/100 - 6 \lambda$ on $[y,p]$, and $600 \lambda/197$ on $[q,y+k]$.
	Set $\lambda = 0.26$
	and $\mu = 0.47$. Update $\ell_1' = \ell_1 + 0.73k$ and $m_1' = m_1 + 0.52k$.
	
	Finally let $\phi = \psi$. Player $3$'s density is: $297/200 - 3 \phi/2$ on $[x,m]$, $300 \phi/97$ on $[n,x+k]$, $3 \phi$ on $[y,p]$, and $(0.99k - \phi k)/(1.97k/3)$ on $[q,y+k]$.
	For $\phi = \psi = 0.25$, the density is in $ \left(1/\sqrt{2}, \sqrt{2}\right)$. Then $\ell_3' = \ell_3 + 0.75k$ and $m_3' = m_3 + 0.5k$. \\
	
	\noindent \emph{Case 2}: $\alpha \in \left(m_2 + 2k/3, m_2 + k\right)$. This time we will hide the interval to the left of the point $x + 2k/3$.
	Define 
	$
	m = x + 2k/3 - 0.01k$, $n = x + 2k/3
	$,$
	p = y + 1.97k/6$, and $q = y + 2.03k/6
	$.
	Set $I' = [m,n]$ and $J' = [p,q]$ and let the densities of all the players be uniform on $I'$ and $J'$. It can be verified that $k/3 + \ell_2 + 1.97k/6 = 3.97k/6 + \ell_2$; moreover its density is uniform on all the new intervals. Update $m_2' = m_2 + 2k/3 - 0.01k$ and $\ell_2' = \ell_2 + 3.97k/6$.
	
	The goal is to fit the valuations of players $1$ and $3$, which implies computing values $0 < \lambda, \mu, \phi, \psi < 0.99$ so that the matrix of valuations is:
			\begin{figure}[h!]
		\centering
		\includegraphics[scale=0.85]{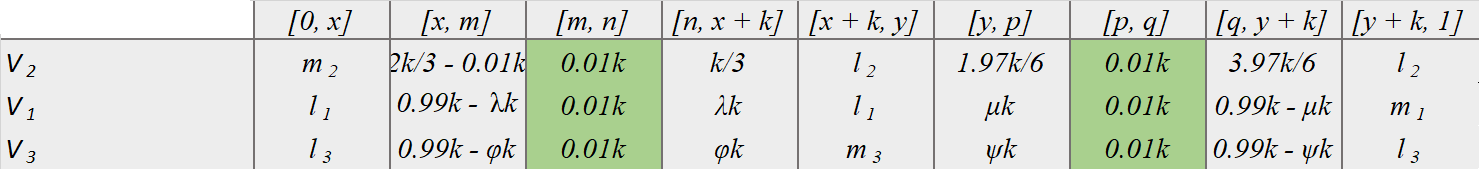}
		\label{fig:partial_2b}
	\end{figure}

	For player $1$, we have $\mu = 0.99 - 2 \lambda$ and density: $(0.99 - \lambda)/(197/300)$ on the interval $[x,m]$, $\lambda k/(k/3) = 3 \lambda $ on $[n,x+k]$, $\mu k/(1.97k/6)= 600/197 \cdot \left(0.99 - 2 \lambda\right)$ on $[y,p]$, and  $(0.99k - \mu k)/(3.97k/6) = 1200 \lambda/297$ on $[q,y+k]$.
	Setting $\lambda = 0.265$ and $\mu = 0.46$ gives the required density bounds. Update $\ell_1' = \ell_1 + 0.725 k$ and $m_1' = m_1 + 0.53k$.
	
	For player $3$ we have $\phi = \psi$ and the densities: $(0.99 - \phi)/(1.97/3)$ on $[x,m]$, $\phi k/(k/3) = 3 \phi $ on $[n,x+k]$, $
	\phi k/(1.97k/6)  = 600 \phi/197$ on $[y,p]$, and $(0.99 - \phi)/(397/600)$ on $[q,y+k]$.
	For $\phi = \psi = 0.25$, the density is in the required range. Set $\ell_3' = \ell_3 + 0.74k$ and $m_3' = m_3 + 0.5k$.
	
	Similarly to part I of the proof, the query asked by the protocol falls outside the new intervals $I'$ and $J'$, and so it can be answered uniformly for player $2$ and proportionally to the weight on the respective interval for players $1$ and $3$. This completes the second part of the proof.
\end{proof}

\newpage
\noindent \textbf{Part c of Lemma \ref{lem:hide}}

\begin{proof} 
	We analyze the situation where the protocol $\mathcal{A}$ addresses a query to player $3$. Let the query be $Cut_3(\alpha)$ and consider two cases: \\
	
	\noindent \emph{Case 1}: $\alpha \in \left(\ell_3, \ell_3 + 2k/3\right]$. We will hide the interval at the right of the point $x + 2k/3$. Define 
	$
	m = x + 2k/3$,$n = x + 2k/3 + 0.01k
	$, $p = y + k/3 - 0.01k$, and $q = y + k/3$.
	
	Set $I' = [m,n]$ and $J' = [p,q]$. Let the value of each player be $0.01k$ for the intervals $I'$ and $J'$. Again all densities are uniform on $I'$ and $J'$. Set the density of player $3$ uniform on all the new intervals and update $\ell_3' = \ell_3 + 2k/3$ and $m_3' = m_3 + 1.94k/3$. The goal is to find appropriate valuations for players $1$ and $2$, that is, appropriate values of the parameters $0 < \lambda, \mu, \phi, \psi < 0.99$ consistent with the next matrix.
			\begin{figure}[h!]
		\centering
		\includegraphics[scale=0.85]{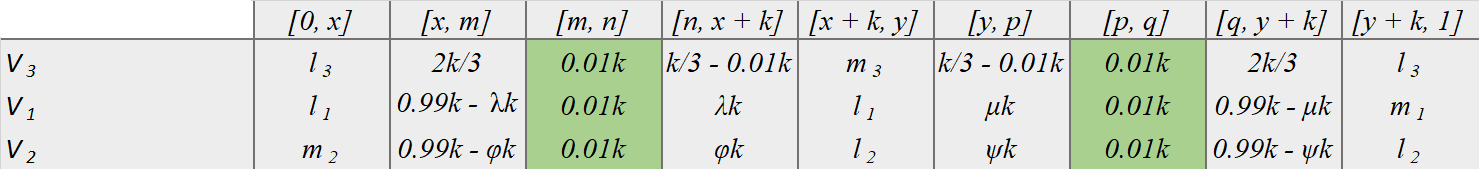}
		\label{fig:partial_3a}
	\end{figure}
	
	Note that 
	$
	\ell_1 + 0.99k - \lambda k = \lambda k + \ell_1 + \mu_k$ if and only if $ \mu = 0.99 - 2 \lambda
	$. Player $1$'s density is $297/200 - 3 \lambda/2$ on $[x,m]$, $300 \lambda/97$ on $[n,x+k]$, $300/97 \cdot \left( 0.99- 2 \lambda \right)$ on $[y,p]$, and $3 \lambda$ on $[q,y+k]$.
	Setting $\lambda = 0.267$ ensures the required density bound. Then $\mu = 0.456$. Update $\ell_1' =\ell_1 + 0.723k$ and $m_1' = m_1 + 0.534k$.
	
	For player $2$,
	$
	\phi k + \ell_2 + \psi k = 0.99k - \psi k + \ell_2$ if and only if $ \psi = 99/200 - \phi/2
	$.
	Player $2$'s density is $297/200 - 3 \phi/2$ on $[x,m]$, $300 \phi/97$ on $[n,x+k]$, $300/97 \cdot \left( 99/200 - \phi/2\right)$ on $[y,p]$, and $297/400 +
	3 \phi/4 $ on $[q,y+k]$.
	Setting $\phi = 0.25$ works. Then $\psi = 0.37$. Update $\ell_2' = \ell_2 + 0.62k$ and $m_2' = m_2 + 0.74k$.\\
	
	\noindent \emph{Case 2}: $\alpha \in \left(\ell_3 + 2k/3, \ell_3 + k\right)$. We will hide the interval at the left of the point
 $x + 2k/3$. Define 
	$
	m = x + 2k/3 - 0.01k$, $n = x + 2k/3
	$, $p = y + k/3$, and $q = y + k/3 + 0.01k$.
	Set $I' = [m,n]$, $J' = [p,q]$.
	
	Let the values of all the players be $0.01k$ for the entire intervals $I'$ and $J'$. Again all densities are uniform on $I'$ and $J'$. Set the density of player $3$ uniform on all the new intervals and update $\ell_3' = \ell_3 + 2k/3 - 0.01k$ and $m_3' = m_3 + 2k/6$. The goal is to find the appropriate values for players $1$ and $2$, or equivalently, $0 < \lambda, \mu, \phi, \psi < 0.99$ as in the next matrix.
			\begin{figure}[h!]
		\centering
		\includegraphics[scale=0.85]{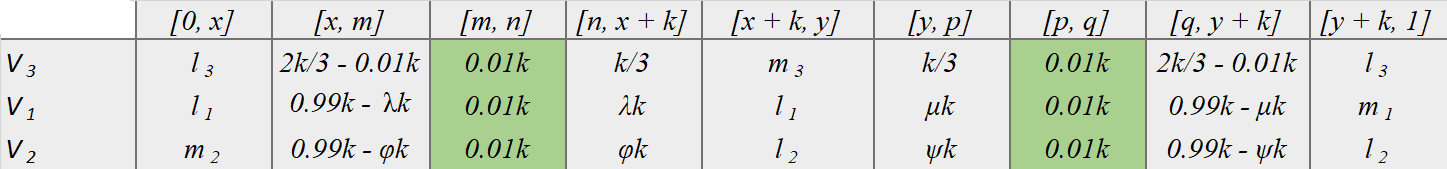}
		\label{fig:partial_3b}
	\end{figure}

	For player $1$ we get $\mu = 0.99 - 2 \lambda$ and density $(297 - 300 \lambda)/197$ on $[x,m]$, $3 \lambda $ on $[n,x+k]$, $
	297/100 - 6 \lambda$ on $[y,p]$, and $600 \lambda/197$ on $[q,y+k]$.
	Setting $\lambda = 0.26$ ensures $2$'s density is in the required range. Then $\mu = 0.47$. Update $\ell_1' = \ell_1 + 0.73k$ and $m_1' = m_1 + 0.52k$.
	
	Finally, we must fit the answers of player $2$. We get $\psi = 99/200 - \phi/2$ and density $(297 - 300 \phi)/197$ on 
	$[x,m]$, $3 \phi$ on 
	$[n,x+k]$, $297/200- 3 \phi/2$ on $[y,p]$, and $(99/200 + \phi/2)/(197/300)$ on $[q,y+k]$.
	Let $\phi = 0.25$. Then $\psi = 0.37$. Update $\ell_2' = \ell_2 + 0.62k$ and $m_2' = m_2 + 0.74k$.
	
	In both cases the query to player $3$ falls outside the interval $I'$, so the query can be answered for all the players in a way that is proportional to their density on the respective interval.
\end{proof}

We can now prove the lower bound.

\begin{proof}(of Theorem \ref{thm:3lb})
	Set the initial configuration to a partial rigid measure system as in the next table, where $k = 0.01$ and $\ell_i = 0.35$,
	$m_i=0.28$ for each player $i$. The initial cuts are at $0.34, 0.35, 0.67, 0.68$, with $I = [0.34,0.35]$ and $J = [0.67,0.68]$. It can be verified that these have the required densities.
	
	\begin{table}[h!]
			\label{tab:3eflb_initial}
		\begin{center}
			\begin{tabular}{l | c | c | c | c | c  r}
				& $[0,0.34]$ & $[0.34, 0.35]$ & $[0.35, 0.67]$ & $[0.67,0.68]$ & $[0.68,1]$\\ \hline
				$V_1$ & 0.35 & 0.01 & 0.35 & 0.01 & 0.28 \\
				$V_2$ & 0.28 & 0.01 & 0.35 & 0.01 & 0.35 \\
				$V_3$ & 0.35 & 0.01 & 0.28 & 0.01 & 0.35 \\
			\end{tabular}
		\end{center}
	\caption{Initial configuration for envy-free lower bound.}
	\end{table}
	
	By iteratively applying Lemma \ref{lem:hide} with every Cut query, a protocol discovers with every cut query a partial rigid system, where the intervals $I$ and $J$ always have uniform density, and their length cannot be diminished by a factor larger than 100 in each iteration.
	By Lemma \ref{lem:density}, if a protocol encounters a partial rigid measure system for which there are no cuts inside $I$ and $J$, where $|I| = |J| = k$, then any configuration attainable with the existing cuts leads to envy of at least $0.01k$.
	To get $\epsilon$-envy, we need $k/100 < \epsilon$, and so the number of queries is $\Omega\left( \log{\frac{1}{\epsilon}}\right)$. 
\end{proof}

\medskip

\noindent \textbf{Theorem} \ref{thm:lb_envyfree_any} (restated)
\emph{Computing a connected $\epsilon$-envy-free allocation for $n\geq 3$ players requires $\Omega \left(\log{\frac{1}{\epsilon}}\right)$ queries.}
\begin{proof}
	For ease of exposition, we assume the number of players is divisible by $3$. Let $K = n/3$ and divide the players in disjoint sets of three, such that each group $S_i = \{3i-2, 3i-1, 3i\}$, for $i \in \{1, \ldots, K\}$, the players in $S_i$ are only interested in the piece
	$J_i = [(i-1)/K,i/K]$, and their valuations form a generalized rigid measure system on $J_i$ with higher densities, such that $K/\sqrt{2} < v_j(x) < K \sqrt{2}$, for each player $j \in S_i$ and $x \in J_i$. By applying Lemma \ref{lem:density} for $a = \sqrt{2}/K$ and $b = K \sqrt{2}$, we get that for any two disjoint pieces $S_1, S_2 \subset J_i$, if the valuation of player $i$ satisfies $V_i(S_1) \geq 2 V_i(S_2)$, then the valuation of another player $j$ in the same group as $i$ satisfies $2 V_j(S_1) \geq V_j(S_2)$. Thus Lemma \ref{lem:far_ef3} still applies for each group $S_i$ and interval $J_i$.
	The queries are handled as follows. Whenever a player $j \in S_i$ receives a cut query outside the piece they are interested in, the answer is given so as to not introduce new cut points. On the other hand, if player $j$ receives a cut query in the interval $S_i$, the answer is given as in the construction of Theorem \ref{thm:3lb}, where the points are scaled to reside in $J_i$. 
	
	Consider the final allocation computed by an RW protocol $\mathcal{A}$, and let $x_i$ be the cut point that separates group $S_i$ from group $S_{i+1}$.
	If $x_1 \geq 1/K$, then the allocation $A$ is $\epsilon$-envy-free among the players in $S_1$ if and only if the algorithm discovers the generalized rigid measure system on $[0,1/K]$, since the piece $[1/K, x_1]$ is worth zero to all the players in $S_1$. Otherwise, $x_1 \leq 1/K$. If $x_2 \geq 2/K$, then for the allocation to be $\epsilon$-envy-free among the players in $S_2$, $\mathcal{A}$ must discover (within $\epsilon$ error) the measure system among the players in $S_2$ on interval $J_2$. Otherwise, we have $x_2 \leq 2/K$. Iteratively, we either find an interval $i \leq K-1$ where the algorithm must solve the problem where the solution is unique among the group $S_i$, or reach $i = K$ with $x_{K-1} \leq (K-1)/K$, case in which $\mathcal{A}$ must find the measure system among the players $S_K$ on $J_K$. Since finding an $\epsilon$-envy-free allocation on $J_i$ among $S_i$ requires $\Omega\left(\log{\frac{1}{\epsilon}}\right)$ queries for all $i \in \{1, \ldots, K\}$, which implies the required lower bound. The cases where the number of players is of the form $n = 3K+1$ and $n = 3K+2$ can be solved by extending the lemmas for three players to four and five players, respectively, by observing that the cases that appear in both Lemma \ref{lem:far_ef3} and Lemma \ref{lem:hide} rely on a number of combinations that are independent of the number of players (in the case of Lemma \ref{lem:far_ef3}, whether a player gets allocated a piece with two columns, one, or none, while in the case of Lemma \ref{lem:hide}, whether the cut falls in an interval worth $m_i$ or $\ell_i$ to a player, and whether the new interval maintained is hidden on the left or right side of the cut).
	Then when $n = 3K+1$, the group $S_1 = \{1,2,3,4\}$, while when $n = 3K+2$, $S_1 = \{1,2,3,4,5\}$.
\end{proof}%

The lower bound of $\Omega \left(\log{\frac{1}{\epsilon}}\right)$ is in fact tight for the class of generalized rigid measure systems, for any (fixed) number of players.

\medskip

\noindent \textbf{Theorem \ref{thm:rms_sim}} (restated).
\emph{
	For the class of generalized rigid measure systems, a connected $\epsilon$-envy-free allocation  can be computed with $O\left( \log{\frac{1}{\epsilon}}\right)$ queries for any fixed number $n$ of players.}
\begin{proof}
	The following moving knife protocol computes an exact envy-free allocation for any number of players with valuations given by a generalized rigid measure system.
		\begin{description} 	
\item[$\bullet$] Let $x_{\alpha} = Cut_{1}(1/n)$ and $x_{\omega} = Cut_{1}(1/2)$. 
	\item[$\bullet$] Player $1$ continuously moves a knife from $x_{\alpha}$ to $x_{\omega}$. For each position $x_1$ of the first knife:
\begin{itemize}
	\item Player $1$ positions a second knife at $x_2 = Cut_1(2 \cdot t_1)$, where $t_1 = V_1(0,x_1)$. Define $s_1 = 1 - 2 \cdot t_1$. 
	\item For each $k = 3 \ldots n-1$, player $1$ positions its $k$-th knife at $x_k = Cut_1(2 \cdot t_1 + (k-2) \cdot s_1)$ if possible, and at $x_k=1$ otherwise.
\end{itemize}		
\item[$\bullet$]If a connected envy-free allocation can be obtained with cuts $x_1, \ldots, x_{n-1}$, player $1$ shouts stop and the cake is divided according to the envy-free allocation.
\end{description}
Since player $1$ goes over all possible choices of $t_1$, if the input is a generalized RMS, there will exist a choice for which the partition demarcated by player $1$ reveals the correct $s_i,t_i$ parameters of all the players. 
	This moving knife protocol can be simulated approximately in the RW model when $n$ is fixed by doing binary search on the parameter $t_1$ for player $1$ and checking for each choice if the tentative allocation is $\epsilon$-envy-free, so computing envy-free allocations for generalized rigid measure systems can be solved with $O \left(\log{\frac{1}{\epsilon}}\right)$ queries.
\end{proof}

\section{Perfect Allocations} \label{apdx:perfect}

In this section we provide the lower and upper bounds for computing perfect allocations for two players. An upper bound of $O(\log{\epsilon^{-1}})$ for this problem can be obtained simulating Austin's moving procedure in the RW model.

\medskip

\noindent \textbf{Theorem \ref{thm:perfect_ub}} (restated)
\emph{An $\epsilon$-perfect allocation for two players can be computed with $O(\log{\frac{1}{\epsilon}})$ queries.}
\begin{proof}
	The main idea is to simulate Austin's moving knife procedure in the RW query model, searching first by the valuation of the first player.
	\medskip
	\begin{quote}
		\textbf{Austin's procedure}:
		\emph{A referee slowly moves a knife from left to right across the cake. At any
			point, a player can call stop.
			When a player called, a second knife is placed at the left edge of the cake.
			The player that shouted stop -- say 1 -- then moves both knives parallel
			to each other.
			While the two knives are moving, player 2 can call stop at any time.
			After 2 called stop, a randomly selected player gets the portion between player
			1's knives, while the other one gets the two outside pieces.
		}
	\end{quote}
	
	In the RW model, we start by asking both players to reveal the midpoint of the cake. If the midpoints coincide within $\epsilon$, we reached an $\epsilon$-perfect allocation. Otherwise, without loss of generality, assume the rightmost midpoint is reported by player $1$ (the case of player $2$ is similar); denote this cut by $z$. Then $V_1(0,z) = V_1(z,1) = 1/2$, while $V_2(0,z) > 1/2$.
	Initialize $w = 1$. 
	
	\begin{figure}[h]
		\centering
		\includegraphics[scale=0.7]{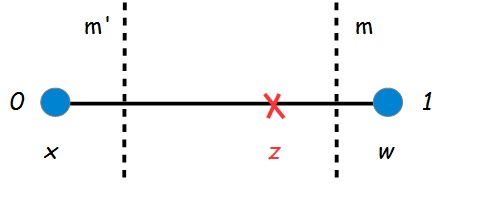}
		\caption{Approximate computation of $\epsilon$-perfect partition. Maintain two points $z$ and $w$, such that the rightmost cut must be situated in the interval $[z,w]$}
		\label{fig:austin}
	\end{figure}
	
	In the RW model we maintain the following invariant:
	
	\begin{enumerate}[(a)]
		\item There exist cut points $0 \leq z < w \leq 1$, such that the piece $[a,z]$ for which $V_1(a,z) = 1/2$ is worth strictly more than $1/2 + \epsilon$ to player $2$, while the piece $[b,w]$ for which $V_1(b,w) = 1/2$ is worth strictly less than $1/2 - \epsilon$ to player $2$.
	\end{enumerate}
	
	Iteratively, given points $w,z$ satisfying property $(a)$, such that $V_1(w,z) \geq \epsilon$, ask player $1$ a Cut query to determine the midpoint $m$ of $[w,z]$, i.e. such that $V_1(z,m) = V_2(m,w)$, and then find through another Cut query the point $m'$ for which $V_1(m',m) = 1/2$. If there exists an $\epsilon$-perfect allocation with cuts $m$ and $m'$ then output it. Otherwise, if $V_2(m',m) > 1/2$, set $z = m$. Else, it must be the case that $V_2(m',m) < 1/2$; set $w = m$.
	
	Each step requires a constant number of queries, and the number of iterations is $O\left( \log{\frac{1}{\epsilon}} \right)$. 
	
	If the interval $[w,z]$ is worth strictly less than $\epsilon$ to player $1$, but an $\epsilon$-perfect allocation has not been found, let $a$ be such that $V_1(a,w) = 1/2$. Any partition with cuts $a$ and $x \in [w,z]$ is $\epsilon$-perfect for player $1$. Then we can search for $ x \in [w,z]$ using the valuation of player $2$, halving the interval $[w,z]$ in each round. A solution is guaranteed to exist and the maximum number of queries addressed to player $2$ is 
	$O\left( \log{\frac{1}{\epsilon}} \right)$.
\end{proof}

As we show next, this bound is optimal.

\medskip

\noindent \textbf{Theorem \ref{thm:perfect_lb}} (restated)
\emph{Computing an $\epsilon$-perfect allocation with the minimum number of cuts for two players requires $\Omega\left(\log\frac{1}{\epsilon}\right)$ queries.} \\

We prove the lower bound by maintaining throughout the execution of any protocol two intervals in which the cuts of the perfect allocation must be situated, such that the distance to a perfect partition cannot decrease too much with any cut query.

\medskip

\noindent \textbf{Lemma} \ref{lem:perfect_induction} (restated) 
\emph{Let $\epsilon > 0$. Consider a two player instance consistent with Figure \ref{fig:perfect_main2},
	where }
\begin{description}
	\item[\hspace{4mm}\emph{1.}] $\epsilon < 0.001 \min\{a,d\}$.
	\item[\hspace{4mm}\emph{2.}] any allocation obtained with cuts $0 < k < \ell < 1$ that is $\epsilon$-perfect from the point of view of player $1$ is worth to player $2$ less than $0.5 - d/100-\epsilon$ when $k < x$ and more than $0.5 + d/100 + \epsilon$ when $k > x+a$.
	\item[\hspace{4mm}\emph{3.}] there are no cut points inside the intervals $I = [x,x+a]$ and $J = [y,y+a]$.
\end{description}
\begin{figure}[h!]
	\centering
	\includegraphics[scale=0.85]{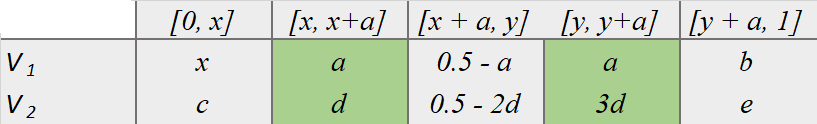}
	\caption{Construction for perfect. Player $1$ has uniform density everywhere; $y= x+0.5$, $0 < a,d \leq 0.1$, $x,b,c,e > 0$, $x + a + b = 0.5$ and $c + 2d + e = 0.5$.}
	\label{fig:perfect_main2}
\end{figure}
\emph{
	Then a new query can be handled so that the valuations remain consistent with Figure \ref{fig:perfect_main2}, such that condition $2$ holds with respect to new intervals $I' = [x',x'+a'] \subseteq I, \;J' = [y',y'+a'] \subseteq J$ and parameters $x',a'=a/100$, $d'=d/100$, and the intervals $I'$ and $J'$ have no cuts inside.}
\begin{proof}
	The valuations can be assumed to be made public outside $I,J$, so queries that fall outside these intervals are handled by the conditions in the lemma. Otherwise, suppose player $2$ receives a query $Cut_2(\alpha)$ in one of the intervals $I,J$. The scenario where player $1$ receives the query will follow from the analysis for player $2$. The new intervals maintained will be $I' = [m,n]$, $J' =[p,q]$, where $m,n,p,q$ are defined depending on the query.
Set $a' = a/100$, $d'=d/8$, and let $0 < k < \ell < 1$ be the defining cuts of a partition that is $\epsilon$-perfect from the point of view of player $1$. Then $0.5 - \epsilon \leq \ell - k \leq 0.5 + \epsilon$. We show that if $k,\ell \not \in I',J'$, then $[k,\ell]$ is either too small or too large for player $2$. 	
		Consider the first scenario, where player $2$'s answer is in the interval $I$. 
	
	\medskip
	
	\noindent \emph{Case 1.a}. $\alpha \in (c,c+d/2]$. Let $m = x+a/2$, $n=x+51a/100$, $p=y+a/2$, $q=y+51a/100$, so the valuations are consistent with the next matrix.

		\begin{figure}[h!]
		\centering
		\includegraphics[scale=0.85]{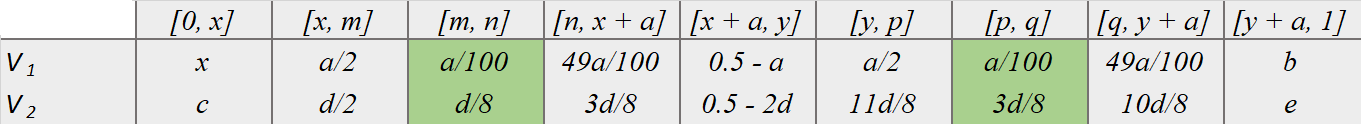}
		\label{fig:perfect_1a}
	\end{figure}
	
	Next we show that if $k \leq m$, then the piece $[k,\ell]$ is worth less than $0.5 - d'/100 - \epsilon$ to player $2$, while if $k > n$, then $[k,\ell]$ is worth more than $0.5 + d'/100 + \epsilon$ to player $2$. If $k < x-\epsilon$, then $\ell \leq y$, and the claim holds by the assumption in the lemma's statement. If $k > x+a+\epsilon$, then $\ell \geq y+a$, and the lemma holds. Otherwise, we have a few cases:
	\begin{itemize}
		\item $k \in [x-\epsilon,x]$. Then $\ell \leq y + \epsilon$. Let $w = V_2(x-\epsilon, x)$. Since $V_2(x-\epsilon,y) = w + V_2(x,y) = w + 0.5-d< 0.5 - d/100 - \epsilon$, we have $w < 99d/100 - \epsilon$. Since $\epsilon < a$, we get
		$$V_2(k, \ell) \leq w + V_2(x,y) + V_2(y,y+\epsilon) \leq (99d/100 - \epsilon)+ (0.5 - d) + \epsilon \cdot \frac{11d/8}{a/2}
		<  0.5 - d'/100 - \epsilon$$
		
		\item $k \in [x,m]$. Let $\delta = k - x \leq a/2$. Then $\ell \leq \min\{p, y + \delta + \epsilon\}$. When $\epsilon < \min\{a,d\}$, we have 
		$$
		V_2(k,\ell) \leq (a/2-\delta) \cdot \frac{d}{a} + 0.5 - 3d/2 + (\delta + \epsilon) \cdot \frac{11d/8}{a/2} = 0.5 - d + \delta \cdot \frac{7d}{4a} \leq 0.5 - d/8 < 0.5 - d'/100 - \epsilon
		$$
		\item $k \in [n,x+a]$. Let $\delta = k - n$. Recall $k+0.5-\epsilon \leq \ell \leq k+0.5+\epsilon$. If $\epsilon > \delta$, then
		\begin{eqnarray*}
			V_2(k,\ell) & \geq & V_2(k,x+a) + V_2(x+a,q)
			\geq \left(3d/8 - \delta \cdot \frac{3d/8}{49a/100}\right) + (0.5 - d/4)  \\
			& \geq & 0.5 + d/8 - \epsilon \cdot \frac{3d/8}{49a/100}
			> 0.5 + d'/100 + \epsilon
		\end{eqnarray*}
		
		Otherwise, $\delta > \epsilon$. We have 
		\begin{eqnarray*}
			V_2(k,\ell) & \geq & V_2(k,x+a) + V_2(x+a,q) + V_2(q,q+\delta - \epsilon) \\
			& \geq & \left(3d/8 - \delta \cdot \frac{3d/8}{49a/100}\right) + (0.5 - d/4) + (\delta - \epsilon) \cdot \frac{10d/8}{49a/100} \\
			& > & 0.5 + d/8 - \epsilon \cdot \frac{10d/8}{49a/100}
			> 0.5 + d'/100 + \epsilon
		\end{eqnarray*}
		\item $k \in [x+a,x+a+\epsilon]$. If $\ell > y+a$, the claim holds. Else, $\ell \in [q,y+a]$. Let $w = V_2(x+a,x+a+\epsilon)$. Then $V_2(x+a+\epsilon,y+a) = 0.5 + d - w > 0.5 + d/100 + \epsilon$, so $w < 99d/100 - \epsilon$. We have $$V_2(k,\ell) \geq V_2(x+a+\epsilon,y+a-\epsilon) \geq 0.5 + d - w - \epsilon \cdot  (10d/8)/(49a/100) > 0.5 + d'/100 + \epsilon$$
	\end{itemize}

	\noindent \emph{Case 1.b}. $\alpha \in (c+d/2,c+d)$. Let $m = x+49a/100$, $n=x+a/2$, $p=y+49a/100$, $q=y+a/2$. Let the valuations outside $I',J'$ be known and consistent with the next matrix.

	\begin{figure}[h!]
	\centering
	\includegraphics[scale=0.85]{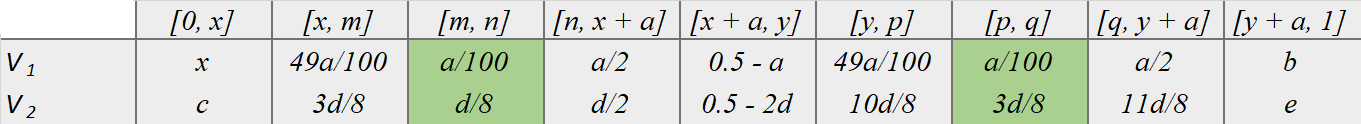}
	\label{fig:perfect_1b}
\end{figure}

	We show the required discrepancy holds for the piece $[k,\ell]$.  If $k < x - \epsilon$ or $k > x+a+\epsilon$, the claim follows by the lemma's condition. The remaining cases are:
	\begin{itemize}
		\item $k \in [x-\epsilon,x]$. We have $w = V_2(x-\epsilon,x) < 99d/100 - \epsilon$. Since $\epsilon < 0.001 a$, we have $V_2(k,\ell) \leq w + V_2(x,y) + \epsilon \cdot (10d/8)/(49a/100) < 0.5 - d'/100 - \epsilon$.
		\item $k \in [x,m]$. Let $\delta = k -x$. Note $V_2(m,y) = 0.5 - 11d/8$. Since $\delta \leq 49a/100 $ and $\epsilon < 0.01 \min\{a,d\}$, we get $$V_2(k,\ell) \leq (49a/100 -\delta) \cdot \frac{3d/8}{49a/100} + 0.5- 11d/8 + (\delta + \epsilon) \cdot \frac{10d/8}{49a/100} < 0.5 - d'/100 - \epsilon 
		$$
		\item $k \in [n,x+a]$. Let $\delta = k -n$. We have $\ell \geq k + 0.5 - \epsilon$ and $\epsilon < 0.01 \min\{a,d\}$. If $\delta \leq \epsilon$, then $$V_2(k,\ell) \geq V_2(n,q) - \epsilon \cdot d/a = 0.5 + d/8 - \epsilon \cdot d/a > 0.5 + d'/100 + \epsilon.$$ Otherwise, $\delta > \epsilon$, thus $\ell \in [q,y+a]$. Since $V_2(x+a,q) = 0.5 - 3d/8$, we get $$
		V_2(k,\ell) \geq d/2 - \delta \cdot d/a + 0.5 - 3d/8 + (\delta - \epsilon) \cdot (11d/8)/(a/2) > 0.5 + d'/100 + \epsilon.
		$$
		\item $k \in [x+a,x+a+\epsilon]$. Using that $w = V_2(x+a,x+a+\epsilon) < 99d/100 - \epsilon$ and $\epsilon < 0.01a$, we get $V_2(k,\ell) \geq 0.5 + d - w - \epsilon \cdot (11d/8)/(a/2) > 0.5 + d'/100 + \epsilon$.
	\end{itemize}
	
	The second scenario, where the answer of player $2$ falls in the interval $J$, has two subcases: 
	
	\medskip
	
	\noindent \emph{Case 2.a}. $\alpha \in (0.5+c-d,0.5+c+d/2]$. Let $m = x+a/2$, $n = x+51a/100$, $p = y + a/2$, $q = y + 51a/100$. Set the valuations consistent with the next matrix.

\begin{figure}[h!]
	\centering
	\includegraphics[scale=0.85]{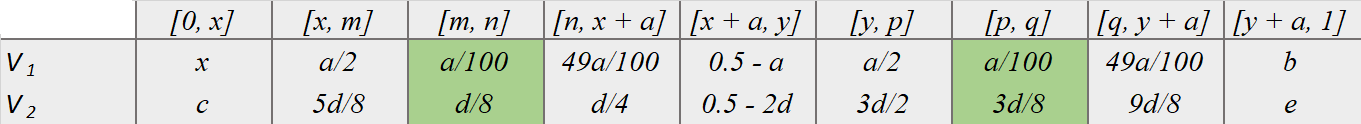}
	\label{fig:perfect_2a}
\end{figure}
	
	If $k < x - \epsilon$ or $k > x+a+\epsilon$, the discrepancy between the valuations of the players for $[k,\ell]$ holds by the lemma's statement. If $k \in [x-\epsilon,x]$, note the change in the density of player $2$ on the interval $[y,p]$ compared to Case 1.a is constant, thus a similar argument works when $\epsilon < 0.001 a$. If $k \in [x+a,x+a+\epsilon]$, the claim also follows from Case 1.a, where the interval $[q,y+a]$ had the same length and higher value density for player $2$ than here. The remaining cases are:
	\begin{itemize}
		\item $k \in [x,m]$. Let $\delta = k - x \leq a/2$. Then $V_2(k,\ell) \leq 0.5 - d + \delta \cdot 7d/(4a) \leq 0.5 - d/8 < 0.5 - d'/100 - \epsilon$.
		\item $k \in [n,x+a]$. Let $\delta = k-n$. If $\delta \leq \epsilon$, the claim follows as in Case 1.a. If $\delta > \epsilon$, then $V_2(k,\ell) \geq 0.5 + d/8 - \epsilon \cdot (9d/8)/(49a/100) > 0.5 + d'/100 + \epsilon$.
	\end{itemize} 
	
	\medskip
	
	\noindent \emph{Case 2.b}. $\alpha \in (0.5+c+d/2,0.5+c+2d)$. Let $m = x+49a/100$, $n = x+a/2$, $p= y+49a/100$, $q = y +a/2$, with valuations as in the matrix on the next page.
	If $k \leq x$, $k\in [x+a,x+a+\epsilon]$, or $k > x+a+\epsilon$, the claim follows as in the previous cases. If 
	\begin{itemize}
		\item $k \in [x,m]$. Let $\delta = k-x \leq 49a/100$. Then $V_2(k,\ell) \leq 0.5 - d - \delta \cdot (d/4)/(49a/100) + (\delta + \epsilon) \cdot (9d/8)/(49a/100) \leq 0.5 -d/8 + \epsilon \cdot (9d/8)/(49a/100) \leq 0.5 - d'/100 - \epsilon$.
		\item $k \in [n,x+a]$. Let $\delta = k-n$. If $\delta \leq \epsilon$ the claim is as in Case 1.b. If $\delta > \epsilon$, $V_2(k,\ell) \geq 0.5 + d/8 + \delta \cdot 290d / (8a) - \epsilon \cdot 300 d / (8a) > 0.5 + d'/100 + \epsilon$.
\begin{figure}[h!]
	\centering
	\includegraphics[scale=0.85]{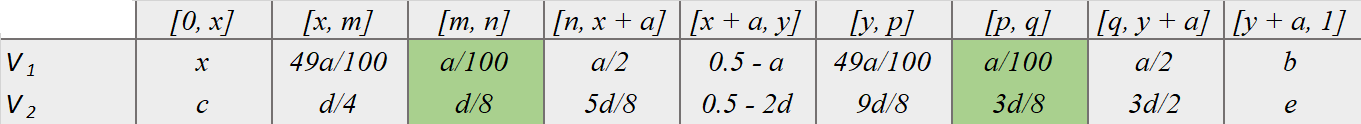}
	\label{fig:perfect_2b}
\end{figure}

	\end{itemize}
	
	Thus each cut query can be answered so that the new partition is still $d'/100$ far from perfect whenever player $1$ believes the middle piece is almost perfect, where the values of $a'$ and $d'$ has been reduced by a constant factor. This completes the proof of the lemma.
\end{proof}

We can now prove the lower bound for perfect allocations.

\medskip

\begin{proof} (of Theorem \ref{thm:perfect_lb})
	Let the initial configuration be defined as follows, with initial parameters $a=d=0.1$, $x=0.2$, $y=0.7$, $b =0.2$, $c=e=0.15$, and intervals $I = [0.2,0.3]$ and $J = [0.7,0.8]$.
	\begin{table}[h!]
		\label{tab:perfect_initial}
		\centering
		\begin{center}
			\begin{tabular}{l | c | c | c | c | c}
				& $[0,0.2]$ & $[0.2,0.3]$ & $[0.3,0.7]$ & $[0.7,0.8]$ & $[0.8,1]$ \\ \hline
				$V_1$ & $0.2$ & $0.1$ & $0.4$ & $0.1$ & $0.2$ \\
				$V_2$ & $0.15$ & $0.1$ & $0.3$ & $0.3$ & $0.15$ \\
			\end{tabular}
		\end{center}
	\end{table}
	Consider any partition $A$ defined by cut points $0 < k < \ell < 1$. Whenever $\epsilon < 0.3$, if $A$ is $\epsilon$ perfect from the point of view of player $1$, then the middle piece is worth 
	\begin{itemize}
		\item less than $0.5 - d/100 - \epsilon$ for player $2$ when $k \leq 0.2$
		\item more than $0.5 + d/100 + \epsilon$ for player $2$ when $k \geq 0.3$.
	\end{itemize}
	
	By iteratively applying Lemma~\ref{lem:perfect_induction} with every cut query received, we obtain that no $\epsilon$-perfect partition can be found as long as $\epsilon < 0.001 \min \{a,d\}$, where the $a$ and $d$ are the values of players $1,2$ for the intervals $I,J$ maintained throughout execution. Since $a,d$ get reduced by a factor of at most $100$ in every iteration, the number of rounds is $\Omega \left( \log{\frac{1}{\epsilon}}\right)$. 
\end{proof}

\section{Moving Knife Protocols} \label{app:movingknife}

A moving knife protocol may have a finite number of ``steps'' where each ``step'' is one of the following: a Cut query, an Eval query, or a Moving Knife step. A moving knife step contains several knives that move continuously across the cake as time passes, as well as several ``triggers'', which are functions of the positions of the knives and become zero when a target configuration has been reached.

\begin{definition} {\bf (An RW Moving Knife Step)}  \label{def:knife}
	There are a constant number $K$ of {\em Devices} some of which have a position on the cake
	and are called 
	{\em Knives} and others can have arbitrary real values and may be called {\em Triggers}.  
	The devices are numbered $1 \ldots K$ and 
	each device $j$ is controlled by a player $i_j$ and has a real value $x_j$ that changes continuously as {\em time} proceeds from $\alpha$ to $\omega$, where $0 \leq \alpha \leq \omega \leq 1$.
	Thus the value of each knife $j$ (i.e. its position on the cake) is given by some function 
	$x_j(t) \in [0,1]$, while the value of each trigger $j$
	is given by a function $x_j(t) \in \Re$.

	The first device is a knife with position equal to time, while
	the value of each additional device,
	$x_j(t) = F_j\left(t, x_1(t), \ldots, x_{j-1}(t) \right)$, may be obtained using at most $\ell$ RW queries to $j$'s owner. \footnote{The values of the devices, once determined, are known to all the players and the center.} The value of a device may depend in an arbitrary way on any information about the players obtained before the moving knife step \footnote{In the RW model, such information means cut points and value labels obtained through Cut queries, Evaluate queries, or moving knife steps executed before the current one.}, but its dependence on the time $t$ and values of previous devices $x_1(t) \ldots x_{j-1}(t)$ is Lipschitz continuous for all hungry valuations.

	An outcome for the moving knife step is the index of a trigger $j$ with different signs at $\alpha$ and $\omega$ (i.e. $x_j(\alpha)\cdot x_j(\omega) \le 0$), together with a time $t$ such that $x_j(t)=0$, as well as the values of all devices $x_{j'}(t)$ at $t$. (If the value of the trigger $x_j$ happens to be monotone then the time $t$ is unique, but in general there may be different such $t$ and any one of them may be given.)
\end{definition}

The information that a protocol retains after executing a moving knife step is the outcome of the step. Moving knife steps cannot always be executed exactly in the RW model, in the sense that it may take infinitely many queries to find a time where a trigger fires,
so we will consider instead approximations. An approximate outcome is such that for any trigger that switches signs from the beginning to the end of time, we get a time and corresponding configuration (i.e. values of all devices) at which the trigger approximately fires.

\begin{definition}[$\epsilon$-Outcome]
	An $\epsilon$-outcome of a moving knife step running from time $\alpha$ to $\omega$ with devices $1 \ldots K$ is
	the index of a trigger $i$ that switches signs from $\alpha$ to $\omega$ (i.e. with $x_i(\alpha)\cdot x_i(\omega) \le 0$), together with a time $t$ and approximate values $\tilde{x}_1(t) \ldots \tilde{x}_K(t)$ of all the devices at this time, so that 
	\begin{itemize}
		\item $x_i(t) \in [-\epsilon, \epsilon]$, and
		\item $|\tilde{x}_j(t) - x_j(t)| \leq \epsilon$, for all $j = 1 \ldots K$. 
	\end{itemize} 
	An $\epsilon$-outcome for an RW query is an answer to the query that is within $\epsilon$ of the correct answer.
\end{definition}

We can now define \emph{moving knife protocols} in the RW model.

\begin{definition}[\textbf{RW Moving Knife Protocol}]
	An RW \emph{moving knife protocol} $\mathcal{M}$ consists of a finite number of steps, each of which is a moving knife step or an RW query. \footnote{We count the answer to an RW query as part of that step.}
	At the end of an execution $\mathcal{M}$ outputs an allocation of the cake
	that depends on the outcomes of the steps it executed.

	We also require robustness: \footnote{Without robustness, moving knife protocols can be very brittle, as can be seen from the following variant of Austin's procedure for $n=2$ players: \emph{
			Player $1$ slides a knife across the cake, from $0$ to $1$. For each position $x$ of the knife, player $1$ positions a second knife at a point $y$ for which $V_1(x,y) = 1/2$. Player $2$ is instructed to shout ``Cut!'' when $V_2(x,y) = 1/2$.
			Given positions $x,y$ of the knives after their movement stops, if $V_1(x,y) = V_2(x,y) = 1/2$, then player $2$ is allocated $[x,y]$ and player $1$ the remainder, else player $1$ receives the whole cake.}
		Casting this procedure in the moving knife framework, we have time running from $0$ to $1$ and three devices, where the first two devices are the knives such that for any time $t \in [0,1]$, their positions are $x_1(t) = t$ and $x_2(t) = x_1(t) + p(t)$, respectively, where $p(t)$ is the smallest point larger than $t$ with $V_2(x_1(t),x_1(t)+p(t)) = 1/2$,
		and the third device is a trigger with value function $x_3(t) = V_2(x,y) - 1/2$.
		If the moving knife step is executed exactly, the procedure always outputs a perfect allocation of the cake.
		The simulation theorem \ref{lem:lipschitzknives} implies that, using $O(\log{\frac{1}{\epsilon}})$ RW queries, we can replace the outcome of the moving knife step with an $\epsilon$-approximate one, consisting of two points $x,y$ for which $V_1(x,y) = 1/2$ and $|V_2(x,y) - 1/2| \leq \epsilon$. But when replacing the exact outcome with the simulated one, we can no longer guarantee that the condition for the first branch is met, and so the execution seen for some valuations will be the second branch, where player $1$ receives the whole cake.}
	if $\mathcal{M}$ outputs $\mathcal{F}$-fair allocations and completes in at most $r$ steps with a partition demarcated by $C$ cuts, then for all $\epsilon > 0$, by iteratively replacing each outcome of a step of $\mathcal{M}$ with an $\epsilon$-outcome \footnote{That is, if the first step is an RW query, then its answer is replaced with an approximate answer to the query, and if it's a moving knife step, it's replaced with an $\epsilon$-outcome of the step. Using this approximate information the protocol then decides what its second step will be, which is also executed approximately, etc.},
	$\mathcal{M}$
	completes in at most $r$ steps with an $\epsilon$-$\mathcal{F}$-fair partition using at most $C$ cuts.
\end{definition}

\subsection{Simulating General Moving Knife Protocols}

\begin{theorem} \label{lem:lipschitzknives}
	Let $\mathcal{M}$ be an RW moving knife step for a cake cutting problem with hungry value densities bounded from above and below by constants.
	
	Then for each $\epsilon > 0$ and every trigger $j$ of $\mathcal{M}$ that switches signs from the start to the end time, we can find an $\epsilon$-outcome associated with this trigger using $O\left(\log \frac{1}{\epsilon}\right)$ RW queries.
	
\end{theorem}
\begin{proof}
	Let $1 \ldots K$ be the devices of $\mathcal{M}$,
	$\alpha$ and $\omega$ be the start and end times, $k$ the Lipschitz constant of the moving knife step, $\Delta$ and $\delta$ the upper and lower bounds on the value densities, where $\Delta > \delta > 0$.
	Since the value densities are bounded by constants, the dependence functions of the moving knives are $k$-Lipschitz, the number of devices is constant and each device is discovered using at most a fixed number $\ell$ of RW queries given the values of previous devices, it follows that there exists a constant 
	$\zeta > \max\{1, k, \Delta\}$ such that for any device $i = 1 \ldots K$ and times $s < t$, where $s,t \in [\alpha, \omega]$, 
	
	\begin{equation} \label{ineq:inductive}
	|x_i(t) - x_i(s)| \leq  \zeta \cdot  |t-s|.
	\end{equation}

	We can simulate the moving knife step $\mathcal{M}$ as follows.
	Initialize the position of the first device (knife) when the time is $\alpha$ and $\omega$ respectively, by setting $a = x_1(\alpha)$ and $b = x_1(\omega)$.
	Let $j$ be any trigger of $\mathcal{M}$, where its values at time $\alpha$ and $\omega$ are $y = x_j(\alpha)$ and $y' = x_j(\omega)$, respectively.
	
	Iteratively, ask player $1$ a Cut query to identify the cut at which the interval $[a,b]$ is halved in its estimation; that is, let $w = (V_1(a) + V_1(b))/2$ and $z = Cut_1(w)$.
	Intuitively this corresponds to checking the configuration when half of the time has elapsed, but since there is no continuous time in the RW model, we use the valuation of player $1$ as a proxy. 
	Since the position of the first knife is equal to time, we have that $\tilde{t} = z$ is the time when $\mathcal{M}$ sets the first knife to position $z$.
	By the conditions in the lemma, given the time $\tilde{t}$ and the position $z$ of the first device at $\tilde{t}$, we can iteratively find using at most $\ell$ RW queries the value of each device $2, \ldots, K$, given the current time and the values of the previous devices. Let $x_j(\tilde{t})$ be the value attained by trigger $j$ at position $z$.
	If $x_j(\tilde{t}) \cdot y_a > 0$, update $y = z$. Else, update $y' = z$. Stop after $2\log \left( \frac{1}{\epsilon'} \right)$ steps, where
	$\epsilon'= \frac{\epsilon \cdot \delta}{\zeta}$.
	
	Since the value of player $1$ for the interval $[a,b]$ halves with each iteration, we get that the final interval $[a,b]$ is very small in player $1$'s estimation: $|V_1(a) - V_1(b)| \leq \epsilon'$. The value densities of the players are bounded from below by $\delta$, thus $|V_1(a) - V_1(b)| \geq \delta |a -b |$, which implies $|a - b| \leq 1/\delta \cdot |V_1(a) - V_1(b)| \leq \epsilon'/\delta$.
	Let $s = a$ and $t = b$ denote the times at which the first knife has positions $a$ and $b$, respectively. Then $|a - b| = |t -s |$.
	By inequality \ref{ineq:inductive}, for each device $i$:
	\begin{eqnarray*}
		|x_i(s) - x_i(t)| & \leq & \zeta \cdot |t - s| \leq \epsilon
	\end{eqnarray*}
	
	This inequality holds in particular for device $j$, case in which $|y - y'| = |x_j(s) - x_j(t)| < \epsilon$.
	There are two cases left. If $y > 0$ and $y'< 0$, then $0 < y \leq y' + \epsilon < \epsilon$, and we can return the time $s$ and the values of the devices at time $s$, $x_1(s) \ldots x_K(s)$. Otherwise, $y < 0$ and $y' > 0$, and so $0 < y' < y + \epsilon < \epsilon$, case in which we return time $t$ and device values $x_1(t) \ldots x_K(t)$.
	Thus for each trigger $j$ we can find an approximate solution with  $O\left(\log \frac{1}{\epsilon}\right)$ RW queries, such that the positions of the other devices are also $\epsilon$-close to their positions at a nearby time where trigger $j$ is exactly zero. This completes the proof.
\end{proof}

We can simulate RW moving knife protocols in the RW model as follows.

\begin{theorem} \label{thm:simulation_robust}[\ref{thm:mainsimulation_robust} in main text]
	Consider a cake cutting problem where the value densities are bounded from above and below by strictly positive constants.
	Let $\mathcal{M}$ be an RW moving knife protocol with at most $r$ steps, such that $\mathcal{M}$ outputs $\mathcal{F}$-fair allocations demarcated by at most a constant number $C$ of cuts.
	
	Then for each $\epsilon > 0$, there is an RW protocol $\mathcal{M}_{\epsilon}$ that uses $O\left(\log \frac{1}{\epsilon}\right)$ queries and computes $\epsilon$-$\mathcal{F}$-fair partitions demarcated with at most $C$ cuts.
\end{theorem}
\begin{proof}
	The proof is immediate given the simulation theorem \ref{lem:lipschitzknives}, which guarantees that we can replace each moving knife step of $\mathcal{M}$ with an $\epsilon$-outcome using $O(\log{\frac{1}{\epsilon}})$ RW queries. Since $\mathcal{F}$ is robust, the approximate fairness of the simulated protocol follows.
\end{proof}

\subsection{Simulating Existing Moving Knife Protocols}

The moving knife procedures from the literature include the Dubins-Spanier procedure (which is in fact equivalent to a discrete RW protocol), Austin's procedure for computing perfect allocations, Austin's extension, which finds for $n=2$ players a partition into $k$ pieces such that each piece is worth $1/k$ to both players, and several procedures for computing envy-free allocations due to Barbanel-Brams, Stromquist, Webb, Brams-Taylor-Zwicker \cite{RW98,BT96} and Saberi-Wang \cite{SW09}.

\begin{theorem} \label{thm:simulation_all}[\ref{thm:mainsimulation_all} in main text]
	The Austin, Austin's extension, Barbanel-Brams, Stromquist, Webb, Brams-Taylor-Zwicker, and Saberi-Wang moving knife procedures
	can be simulated with $O\left(\log{\frac{1}{ \epsilon}}\right)$ RW queries when the value densities are 
	bounded from above and below by positive constants.
\end{theorem}
\begin{proof}
	Let $\Delta \geq \delta > 0$ be the upper and lower bounds on the valuations, respectively.
	First note that all these procedures have a constant number of players, their moving knife steps have a constant number of devices, and the first device is a knife held by the referee, the position of which is equal to the time. The dependence functions of the known protocols are in the case of the triggers, simple linear functions (that trivially meet the Lipschitz condition) of the form $V_i(S) - V_i(T)$ for some player $i$ and pieces $S,T$, and
	in the case of each knife, the outcomes of a Cut query, where the value given to the query is equal to the value of a player for some demarcated interval.
	Finally, we can reach within $\epsilon/\Delta$ the same positions of the knives as in the continuous procedure, thus approximating within $\epsilon$ the value of each player for every piece demarcated by two adjacent knives. 
	\medskip
	
	$\bullet$ \emph{Austin's procedure}. We will use several steps, each of which is an RW query or a moving knife step, as follows.
	The first step of Austin's procedure is discrete, and can be implemented with a Cut query in the RW model, in which player $1$ is asked to cut the cake in half; let $z = Cut_1(0.5)$. Then ask player $2$ to evaluate the generated piece: $\alpha = Eval_2(z)$. If $\alpha = 0.5$, we have found an exact perfect allocation. 
	Otherwise, w.l.o.g. assume $V_2(0,z) < 0.5$ (the other case is similar).
	We construct a moving knife step with three devices, where the first device will be a knife and represents the value of the referee knife. Its value (i.e. position) is equal to the time that has elapsed as long as the time is less than $z$, and equal to $z$ otherwise.
	The second device is also a knife, and its value can be obtained with an Eval and Cut query to player $1$ as follows: let $v = Eval_1(z)$ and set the value of the second knife to $p = Cut_1(0.5+v)$.
	The third device is a trigger and its value is equal to $V_1(z,p)$, which can be obtained with a single evaluation query to player $2$, $Eval_2(z,p)$.
	Note that $K = 3$ and $\ell = 2$, where $K$ is the number of devices and $\ell$ the number of queries required to find the value of a device given the values of the previous devices. By Lemma \ref{lem:lipschitzknives}, we can simulate the procedure in the RW model to find a configuration where the trigger has value zero with $O\left(\log{\frac{1}{\epsilon}}\right)$ queries, which corresponds to an $\epsilon$-perfect allocation.
	\medskip
	
	$\bullet$ \emph{Austin's extension} \cite{BT96,RW98}. This moving knife procedure is an extension of Austin's procedure for $n=2$ players, which finds a partition into $k$ pieces $A = (A_1, \ldots, A_k)$ into $k$ pieces, such that $V_i(A_j) = 1/k$ for each $i = 1,2$, $k = 1 \ldots k$. The procedure works as follows. Initialize $\ell = k$. \emph{Alice makes $\ell - 1$ cuts that divide the cake into $\ell$ intervals, each worth $1/k$ to Alice. Then there is a piece that Bob values at (weakly) less than $1/k$, and an adjacent piece that Bob values at (weakly) more than $1/k$. 
		Then the referee places a knife at the boundary of one of the pieces and Alice places a second knife so that the interval between the two knives is worth $1/k$ to her. The knives are moved continuously, keeping Alice's value for the piece between them at $/1k$, until meeting the cut points of the other piece. By the intermediate value theorem, there is a point where Bob agrees that the piece between the knives is worth exactly $1/k$. Append this piece to the final allocation, update $\ell = \ell - 1$, and repeat the procedure on the remainder of the cake.} 
	
	For fixed $k$, this procedure contains a constant number of discrete RW queries and moving knife steps that are very similar to Austin's procedure, which can be simulated approximately.
	
	\medskip
	
	$\bullet$ \emph{Barbanel-Brams procedure}. The first steps of this procedure consist of RW queries, which are already in the RW model. Afterwards, we branch in two cases, each of which is a moving knife step. For each such step, we keep track of six devices. We illustrate case 1:  the first device is the knife (corresponding to the referee's knife) with position $x_1(t)= t$, the second device is also a knife, the position $x_2(t)$ of which is determined by asking player $1$ a Cut query, the third device is a trigger with valuation determined by player $2$'s estimation of the difference between pieces $[0,x_1(t)]$ and $[x_1(t),x_2(t)]$ (i.e. $V_2(0,x_1(t)) - V_1(x_1(t),x_2(t))$), the fourth device is a trigger with valuation $V_2(x_1(t),x_2(t)) - V_2(x_2(t),1)$, the fifth and sixth devices are triggers similar to triggers three and four, with the difference that their values are determined by player $3$ instead of player $2$.
	
	\medskip
	
	$\bullet$ \emph{Webb, Brams-Taylor-Zwicker, Saberi-Wang} \cite{BT96,RW98,SW09}:
	The Webb procedure computes an envy-free allocation among $n=3$ players and uses as subroutines the Dubins-Spanier protocol, which can be implemented exactly in the RW model, and Austin's extension, which can be simulated approximately as discussed. 
	The Brams-Taylor-Zwicker protocol uses two calls to Austin's procedure, which can be simulated approximately since the protocol is robust; the same is true of the Saberi-Wang procedure, which also uses Austin's procedure as a subroutine.
	
	\medskip
	
	$\bullet$ \emph{Stromquist's procedure} \cite{RW98}: In this case we will need more complex dependence functions. The procedure is as follows:
	\emph{
		\begin{quote}
			A referee continuously moves a knife from $0$ to $1$. For each position $x$ of the referee knife, every player $i$ positions their own knife at their midpoint $m_i$ of the piece $[x,1]$. Denote $y = \textsc{median}(m_1, m_2,m_3)$. Each player $i$ is instructed to shout stop when $V_i(0,x) = \max\{V_i(x,y),V_i(y,1)\}$. When that happens, the cake is allocated as follows: the player $i$ that called stop receives the piece $[0,x]$, the player $j$ with the  mark $m_j \leq y$ receives the piece $[x,y]$, while the remaining player receives the piece $[y,1]$.
		\end{quote}
	}

	We can implement the procedure as a single moving knife step in our framework using four knives and four triggers as follows. The first knife has a position $x_1(t)$ equals to time and runs from $0$ to $1$. Each knife $k = 2,3,4$ is positioned by player $k-1$ as a result of a cut query $Cut_{k-1}(\alpha_{k-1} + \beta_{k-1})$, where $\alpha_{k-1} = Eval_{k-1}(x_1(t))$ and $\beta_{k-1} = \frac{1 - \alpha_{k-1}}{2}$.
	
	The first three triggers are devices $5,6,7$,
	one for each player $k = 1,2,3$, such that  
	$x_{k+4}(t) = V_k(0,x_1(t)) - \max\left\{V_k\left(x_1(t),y\right), V_k(y,1)\right\}$.
	The fourth trigger is the eighth device with the following dependence function:
	\begin{description}
		\item[\hspace{4mm}(a)] If $x_i(t) < 0$ for $i=5 \ldots 7$, let $x_8(t) = \min_{i=5 \ldots 7}\{ x_i(t)\}$.
		\item[\hspace{4mm}(b)] Else, if $x_i(t) > 0$ for $i = 5 \ldots 7$, sort the triggers by $j,k,\ell$, such that $x_j(t) \geq x_k(t) \geq x_{\ell}(t)$. Define 
		$x_8(t) = x_k(t) + x_{\ell}(t)$.
		\item[\hspace{4mm}(c)] Else, if $x_i(t) \leq 0$ for some trigger $i \in \{5,6,7\}$ and $x_j(t) \geq 0$ for all $j\in \{5,6,7\} \setminus \{i\}$, define 
		$x_8(t) = \min_{j \in \{5,6,7\} \setminus \{i\}} x_j(t)$.
		\item[\hspace{4mm}(d)] Else, let $x_8(t) = 0$. (In this case, we have
		$x_i(t) \geq 0$ for some trigger $i$ and $x_j(t) \leq 0$ for all $j \in \{5,6,7\} \setminus \{i\}$.)
	\end{description}
	
	Clearly the positions of the knives are continuous in time since the median function is continuous. The value of the first triggers are simple linear functions of the valuations for the pieces demarcated by knives and are continuous. The value of the fourth trigger can also be seen to be continuous by checking the combinations of switching points indicated by the four cases in its definition. Moreover, we have that $x_8(0) = -0.5 < 0$ and $x_8(1) = 2 > 0$. Thus there exists a time when the trigger is exactly zero.
	To simulate this moving knife step, the search strategy is that when the trigger value is negative, we must move the referee knife to the right, and when the trigger value is positive, to the left. 
	
	Now suppose we have a time $t$ at which the eighth trigger is approximately zero, i.e. $-\epsilon \leq x_8(t) \leq \epsilon$. Let $y$ be the median at this time. If case $(a)$ holds at time $t$, then there exists a player $i$ such that $-\epsilon \leq w_i - V_i(0,x_1(t)) \leq \epsilon$. Give player $i$ the piece $[0,x_1(t)]$. Let $j \neq i$ be the player whose knife is less than or equal to the median and $k$ the remaining player, whose knife is greater than or equal to the median.
	Then give player $j$ the piece $[x_1(t), y]$ and player $k$ the piece $[y,1]$.
	
	In case $(b)$, since $V_i(0,x_1(t)) > w_i$ for all $i$ and $x_8(t) \in [-\epsilon,\epsilon]$, we have $|V_k(0,x_1(t)) - w_k| \leq \epsilon$ and 
	$|V_{\ell}(0,x_1(t)) - w_{\ell}| \leq \epsilon$. W.l.o.g, let $j$ be the player in $\{j,k\}$ whose knife is less than or equal to the median. Then give $[0,x_1(t)]$ to player $k$, $[x_1(t),y]$ to player $j$, and $[y,1]$ to player $\ell$.
	
	In case $(c)$, let $j \neq i$ be the player with an approximate tie (i.e. $|V_j(0,x_1(t)) - w_j| \leq \epsilon$) and $k \neq i,j$ the remaining player, with $V_k(0,x_1(t)) \geq w_k$. Let $k$ take the piece $[0,x_1(t)]$. If the position of player $i$'s knife is to the left of $y$, the allocation in which player $i$ takes the piece $[x_1(t),y]$ and player $j$ takes $[y,1]$ is envy-free. If player $i$'s knife is located at $y$, let player $j$ pick its favorite between $[x_1(t), y]$ and 
	$[y,1]$, while player $k$ takes the remainder.
	If player $i$'s knife is located to the right of $y$, let player $j$ take $[x_1(t), y]$ and player $i$ the piece $[y,1]$.
	
	In case $(d)$, give player $i$ the piece $[0,x_1(t)]$. An envy-free allocation can be obtained by giving the piece $[x_1(t), y]$ to the player whose knife is less than or equal to $y$, and the remaining player the piece $[y,1]$. This completes the proof.
\end{proof}

Note that our previous simulations of the Barbanel-Brams and Austin procedures do not require that the valuations are bounded at all, since we use the specific formulations of the value functions of the triggers to eliminate the players one by one from consideration. This is not necessarily possible when the dependence functions of the devices are more complex.

\bigskip

\noindent \textit{\textbf{An Equitable Protocol}}. 
Next we show a simple moving knife step in the Robertson-Webb model for computing equitable allocations for any number of hungry players, which can be implemented with $n+1$ devices, one of which is a trigger and the remainder of $n$ are knives. This moving knife step illustrates a trigger whose valuation function depends on time and the physical location of the knives (rather than the values of the players for the pieces demarcated by the knives, as is common in the previous examples).

A more complex moving knife procedure for computing exact equitable allocations that is not in the RW model but works even when the valuations are not hungry was discovered independently by \cite{SS17}. Their protocol is based on the same idea of multiple knives moving in parallel and has more steps designed to detect regions valued at zero by some players. Note the existence of connected equitable allocations for any order of the players was established by Cechlarova, Dobos, and Pillarova \cite{CDP13}, with a simplified proof based on the Borsuk-Ulam theorem given by Cheze \cite{Cheze17}.

\medskip
\begin{quote} 
	\emph{\textbf{Equitable Protocol}} : 
	{\em Player $1$ slides a knife continuously across the cake, from $0$ to $1$. 
		For each position $x_1$ of the knife, player $1$ is asked for its value of the piece $[0,x_1]$; then 
		each player $i = 2 \ldots n$ iteratively positions its own knife at a point $x_i \in [x_{i-1}, 1]$ with $V_i(x_{i-1}, x_{i})$ $=$ $V_1(0,x_1)$ if possible, and at $x_i = 1$ otherwise. 
		
		Player $n$ shouts ``Stop!" when its own knife reaches the right endpoint of the cake (i.e., $x_n = 1$).
		The cake is allocated in the order $1 \ldots n$, with cuts at $x_1 \ldots x_{n-1}$.
	}
\end{quote}

\begin{theorem} \label{thm:appequit_step}[\ref{thm:equit_step} in main text]
There is an RW moving knife protocol that computes connected equitable allocation for any number of $n$ hungry players.
\end{theorem}
\begin{proof}
	We first note the protocol consists in fact a single moving knife step.
	We first note that given the position of the first knife, the value of player $1$ for the piece $[0,x_1]$, $a = V_1(0,x_1)$, can be obtained with an Eval query, and the position of each knife $j = 2 \ldots n$ can be obtained from the position of knife $j-1$ by using an Eval query addressed to player $j$ for the piece $[0,x_{j-1}]$, i.e. $b := Eval_j(x_{j-1})$, followed by a Cut query:  $Cut_j(a+b)$. Also, at every point in time, the knives are positioned as to give an equitable allocation on the subset $[0,x_n]$ of the cake.
	
	To implement this moving knife step in the framework of Definition \ref{def:knife}, we will use $n+1$ devices and time running from $\alpha=0$ to $\omega=1$. The first $n$ devices are knives, in the order $1 \ldots n$, such that the position of knife $i$ depends on knives $1 \ldots i-1$, while device $n+1$ will be a trigger with values of different signs at time $0$ and $1$, respectively, and that will fire (i.e. turn zero) when an exact equitable allocation is reached. 
	To define the value function for the trigger, consider a (hypothetical) augmented resource, of length $n$, where the valuations of the players on $[1,n]$ are uniform and on $[0,1]$ are as defined in the actual instance. 
	From the Cut and Eval queries in the actual moving knife step, the center can maintain throughout time the positions of some fictitious knives on the hypothetical cake, $y_1(t) \ldots y_n(t)$, which are always positioned in an equitable configuration; that is, $y_1(t) = x_1(t)$, while $y_2(t) \ldots y_n(t)$ are iteratively set such that $V_1(0,y_1(t)) = V_i(y_{i-1}(t),y_i(t))$. Then at time $t=0$ we have $y_1(0) = \ldots = y_n(0) = 0$, while at time $t=1$ we have $y_1(1) = 1, y_2(1) = 2, \ldots, y_n(1) = n$. Define the value function of the trigger as $f(t) = y_n(t) - 1$. Then $f(0) = -1$ and $f(1) =n-1$. Since $f(t)$ is continuous, by the intermediate value theorem there exists a time $t$ when $f(t) = 0$, which corresponds to $y_n(t) = 1$. 
	Thus the protocol is a moving knife step as defined in Definition \ref{def:knife} and computes a connected equitable allocation of the whole cake. This completes the proof.
\end{proof}

\section{The Stronger and Weaker RW Models} \label{app:weakstrong}
We also formalize two other query models. The first one, which we call $RW^+$, differs in that the inputs to evaluate queries need not be previous cut points and the protocol can use arbitrary points to demarcate the final allocation. Our lower bounds carry over for this model.

\begin{definition}[$RW^+$ query model]
	An $RW^{+}$ protocol for cake cutting communicates with the players via two types of queries:
	\begin{itemize}
		\item{}$\emph{\textbf{Cut}}_i(\alpha)$: Player $i$ cuts the cake at a point $y$ where $V_{i}([0,y]) = \alpha$, for any $\alpha \in [0,1]$.
		\item{}$\emph{\textbf{Eval}}_i(y)$: Player $i$ returns $V_{i}([0,y])$, for any $y \in [0,1]$.
	\end{itemize}
	At the end of execution an $RW^{+}$ protocol outputs an allocation that can be demarcated by any points (regardless of whether they are previous cut points or not).
\end{definition}

Another natural model, that we call $RW^{-}$, has only one type of query that the protocol can ask the players.

\begin{definition}[$RW^-$ query model]
	An $RW^{-}$ protocol for cake cutting communicates with the players via queries of the form
	\begin{itemize}
		\item $\emph{\textbf{Eval}}_i(y)$: Player $i$ returns $V_{i}([0,y])$, where $y \in [0,1]$ is arbitrarily chosen by the center.
	\end{itemize}
	At the end an $RW^-$ protocol outputs an allocation that can be demarcated by any points.
\end{definition}

The upper and lower bounds for envy-free, perfect, and equitable allocations also hold in the $RW^+$ model as a corollary of the proofs in the RW model.

\begin{corollary}
	Computing a connected $\epsilon$-envy-free allocation for $n\geq 3$ players in the $RW^+$ model requires $\Omega \left(\log{\frac{1}{\epsilon}}\right)$ queries.
\end{corollary}
\begin{proof}
	Let $\epsilon > 0$ and $\mathcal{P}$ a protocol in the $RW^+$ model.
	The main observation that allows us to conclude the extension of the bound to the $RW^+$ model is that in Lemma \ref{lem:hide}, if during the execution of $\mathcal{P}$ a player gets an Evaluate query $Eval_i(y)$ for some arbitrary point $y$, then
	\begin{itemize}
		\item 	if $y$ is outside the hidden intervals (in which the cuts of the envy-free allocation reside), the query can be answered in the same way as in the RW model, and 
		\item if $y$ is inside one of the hidden intervals, then it can be handled as a $Cut_i(y)$ query in the RW model. Thus any $RW^{+}$ protocol $\mathcal{P}$ does not manage to reduce the length of the hidden interval by more than a factor of $100$ with each query. 
	\end{itemize}
	
	Denote by $I_1 \ldots I_{n-1}$ the hidden intervals (in which the envy-free cut points lie) after $\mathcal{P}$ has finished issuing queries, where $I_j = [x_j, y_j]$ for all $j$. By the way the intervals are maintained during $\mathcal{P}$'s execution, there are no cuts inside any $I_j$. Suppose $\mathcal{P}$ uses a cut $x^* \in (x_j, y_j)$ to demarcate the pieces of two players adjacent in the final allocation, $A$. Then an adversary can set $V_i([x_j, x^*]) = 0$. Thus the only way that $A$ can be $\epsilon$-envy-free is if there exists an $\epsilon$-envy-free allocation $A'$ that does not use cuts inside $I_j$. Thus an $RW^+$ protocol has no advantage over an $RW$ protocol.
\end{proof}

\begin{corollary}
	Computing an $\epsilon$-perfect allocation for $n=2$ players in the $RW^+$ model requires $\Omega \left(\log{\frac{1}{\epsilon}}\right)$ queries.
\end{corollary}
\begin{proof}
	The proof follows from the observation that Lemma \ref{lem:perfect_induction} holds in the $RW^+$ model too. If the uniform intervals maintained in the lemma are $I = [x,x+a]$ and $J = [y,y+a]$, with $y = x+0.5$, then any evaluate query outside the intervals $I,J$ can be handled in the usual way (consistent with history). If an evaluate query is received inside one of $I,J$, let $Eval_i(s)$ denote the query. If $s \in (x,x+a/2]$, this can be handled as Case $1.a$ in Lemma \ref{lem:perfect_induction}, while if $s \in (x+a/2, x)$, it can be handled as Case 1.b in Lemma \ref{lem:perfect_induction}. Also note that if the final interval $I$ is not worth less than $\epsilon$ to both players, an $\epsilon$-perfect allocation cannot be found using additional arbitrary cuts inside $I,J$.
\end{proof}

\begin{corollary}
	Computing an $\epsilon$-equitable allocation for $n=2$ players in the $RW^+$ model requires $\Omega \left(\log{\frac{1}{\epsilon}}\right)$ queries.
\end{corollary}
\begin{proof}
	The proof follows from the observation that Lemma \ref{lem:ep_distance} also holds in the $RW^+$ model, by noting that any evaluate query $Eval_i(s)$, where $s$ is not a previous cut point, can be handled in the same way as a Cut query.
\end{proof}

In the $RW^{+}$ model, the proofs for simulating moving knife steps and procedures are the same except that we do not need the valuations to be bounded from below since the center can reduce (half) the time directly with each iteration (instead of reducing it through the lens of the players' valuations). We get the simulations:

\smallskip

\begin{theorem}
	Let $\mathcal{M}$ be an $RW^+$ moving knife step for a cake cutting problem where the players have hungry value densities upper bounded by a constant $D > 0$.
	
	Then for each $\epsilon > 0$ and every trigger $j$ of $\mathcal{M}$ that switches signs from the beginning to the end of time that $\mathcal{M}$ runs, we can find an $\epsilon$-outcome associated with trigger $j$ using $O\left(\log \frac{1}{\epsilon}\right)$ $RW^+$ queries.
\end{theorem}

\begin{theorem} \label{thm:simulation_robust_rw+2}
	Consider a cake cutting problem where the value densities are bounded from above by constant $D>0$.
	Let $\mathcal{M}$ be an $RW^+$ moving knife protocol with at most $r$ steps, such that $\mathcal{M}$ outputs $\mathcal{F}$-fair allocations demarcated by at most a constant number $C$ of cuts.
	
	Then for each $\epsilon > 0$, there is an $RW^+$ protocol $\mathcal{M}_{\epsilon}$ that uses $O\left(\log \frac{1}{\epsilon}\right)$ queries and computes $\epsilon$-$\mathcal{F}$-fair partitions demarcated with at most $C$ cuts.
\end{theorem}

\medskip
We also introduce a weaker model, which we call $RW^{-}$, where the protocol can ask the players only the evaluate type of query.

\begin{definition}[$RW^-$ query model]
	An $RW^{-}$ protocol for cake cutting communicates with the players via queries of the form
	\begin{itemize}
		\item $\emph{\textbf{Eval}}_i(y)$: Player $i$ returns $V_{i}([0,y])$, where $y \in [0,1]$ is arbitrarily chosen by the center.
	\end{itemize}
	At the end an $RW^-$ protocol outputs an allocation that can be demarcated by any points.
\end{definition}

If the valuations are arbitrary, then an $RW^-$ protocol may be unable to find any fair allocation at all. The reason is that no matter what queries an $RW^-$ protocol asks, one can hide the entire instance in a small interval that has value $1$ for all the players; this interval will shrink as more queries are issued, but can be set to remain of non-zero length until the end of execution. 

\medskip

However, if the valuations are bounded from above\footnote{There exist other types of valuations on which the $RW^-$ model may be useful, such as piecewise constant valuations defined on a grid, with the demarcations between intervals of different height known to the protocol.}, then an $RW^-$ protocol is quite powerful.

\begin{theorem}
	Suppose the valuations of the players are bounded from above by a constant $D > 0$. Then any $RW^+$ query can be answered within $\epsilon$-error using $O(\log{\frac{1}{\epsilon}})$ $RW^-$ queries. 
\end{theorem}
\begin{proof}
	Let there be an instance with arbitrary valuations $v_1 \ldots v_n$ such that $v_i(x) < D$ for all $x \in [0,1]$ and $i \in N$.
	Since an $RW^-$ protocol can use the same type of evaluate queries as an $RW^+$ protocol, the simulation has to handle the case where the incoming query is a cut. Let this be $Cut_i(\alpha)$ and denote by $x$ the correct answer to the query. In order to find an approximate answer using only evaluate queries, initialize $\ell = 0$, $r=1$, and search for the correct answer:
	\begin{enumerate}
		\item Let $m = (\ell + r)/2$.
		\item Ask player $i$ the query $Eval_i(m)$ and let $w$ be the answer given. If $|w - \alpha| \leq \epsilon$, return $m$. 
		\item Otherwise, if $m > \alpha$, set $r = m$, and if $m < \alpha$, set $\ell = m$; return to step 1.
	\end{enumerate}
\end{proof}

\end{document}